\documentclass[journal]{IEEEtran}

\usepackage{amsmath}
\usepackage{amsfonts}
\usepackage{amssymb}
\usepackage{amsthm}
\usepackage{tabularx}
\usepackage{mathrsfs} 
\usepackage{mathtools}
\usepackage{multirow}
\usepackage{caption}

\usepackage{url}
\usepackage[colorlinks]{hyperref}

\newtheorem{theorem}{Theorem}
\newtheorem{remark}{Remark}

\newtheorem{Proposition}{Proposition}
\newtheorem{lemma}{Lemma}
\newtheorem{corollary}{Corollary}

\usepackage{stfloats}
\usepackage{float}
\usepackage{graphicx}
\usepackage{xcolor}
\usepackage{subfigure}

\makeatletter
\def\blfootnote{\xdef\@thefnmark{}\@footnotetext}
\makeatother

\begin{document}
	
\title{Physical Layer Security over Fluid Antenna Systems: Secrecy Performance Analysis}

\author{Farshad~Rostami~Ghadi, \IEEEmembership{Member}, \textit{IEEE}, Kai-Kit~Wong, \IEEEmembership{Fellow}, \textit{IEEE},\\
F.~Javier~L\'opez-Mart\'inez,~\IEEEmembership{Senior Member},~\textit{IEEE},~Wee Kiat New,~\IEEEmembership{Member}, \textit{IEEE},\\
Hao Xu,~\IEEEmembership{Member}, \textit{IEEE}, and~Chan-Byoung~Chae, \IEEEmembership{Fellow},~\textit{IEEE}
\vspace{-10mm}
}
\maketitle

\begin{abstract}
This paper investigates the performance of physical layer security (PLS) in fluid antenna-aided communication systems under arbitrary correlated fading channels. In particular, it is considered that a single fixed-antenna transmitter aims to send confidential information to a legitimate receiver equipped with a planar fluid antenna system (FAS), while an eavesdropper, also taking advantage of a planar FAS, attempts to decode the desired message. For this scenario, we first present analytical expressions of the equivalent channel distributions at the legitimate user and eavesdropper by using copula, so that the obtained analytical results are valid for any arbitrarily correlated fading distributions. Then, with the help of Gauss-Laguerre quadrature, we derive compact analytical expressions for the average secrecy capacity (ASC), the secrecy outage probability (SOP), and the secrecy energy efficiency (SEE) for the FAS wiretap channel. Moreover, for exemplary purposes, we also obtain the compact expression of ASC, SOP, and SEE by utilizing the Gaussian copula under correlated Rayleigh fading channels as a special case. Eventually, numerical results indicate that applying the fluid antenna with only one activated port to PLS can guarantee more secure and reliable transmission, when compared to traditional antenna systems (TAS) exploiting maximal ratio combining (MRC)  and antenna selection (AS) under selection combining (SC).
\end{abstract}


\blfootnote{The work of F. Rostami Ghadi, W. K. New and K. K. Wong is supported by the Engineering and Physical Sciences Research Council (EPSRC) under Grant EP/W026813/1. The work of H. Xu is supported by the European Union's Horizon 2020 Research and Innovation Programme under Marie Sk$\L$odowska-Curie Grant No. 101024636. The work of F. J. L\'opez-Mart\'inez  is funded in part by Junta de Andaluc\'ia through grant EMERGIA20-00297, and in part by MCIN/AEI/10.13039/501100011033 through grant PID2020-118139RB-I00. The work of C.-B. Chae is supported by the Institute for Information and Communication Technology Promotion (IITP) grant funded by the Ministry of Science (MSIT) and Information and Communication Technology (ICT), Korea (No. 2024-00428780, No. 2021-0-00486). 
}
\blfootnote{\noindent Farshad Rostami Ghadi, Kai-Kit Wong, Wee Kiat New, and Hao Xu are with the Department of Electronic and Electrical Engineering, University College London, WC1E 6BT London, UK. K. K. Wong is also affiliated with Yonsei Frontier Laboratory, Yonsei University, Seoul 03722, Korea. (E-mail:$\{\rm f.rostamighadi,kai\text{-}kit.wong,a.new,hao.xu\}@ucl.ac.uk$).}
\blfootnote{\noindent F.J. L\'opez-Mart\'inez is with the Department of Signal Theory, Networking and Communications, Research Centre for Information and Communication Technologies (CITIC-UGR), University of Granada, 18071, Granada (Spain), and also with the Communications and Signal Processing Lab, Telecommunication Research Institute (TELMA), Universidad de M\'alaga, M\'alaga, 29010, (Spain). (E-mail: $\rm fjlm@ugr.es$).}
\blfootnote{\noindent C.-B. Chae is with the School of Integrated Technology, Yonsei University, Seoul 03722 Korea. (E-mail:$\rm cbchae@yonsei.ac.kr$).}
\blfootnote{Corresponding author: Kai-Kit Wong.}

\vspace{-3mm}
\section{Introduction}\vspace{0mm}\label{sec-intro}
\IEEEPARstart{D}{ue to the} explosive growth in the use of intelligent devices such as smart phones/tablets and consequently the importance of massive connectivity in wireless networks, it is expected that the sixth-generation (6G) mobile technology will provide more efficient and reliable transmission compared with the present fifth-generation (5G) mobile communication system \cite{tariq2020speculative}. The current 5G wireless technology is based on exploiting the multiple-input multiple-output (MIMO) system, which has advanced into multiuser and massive MIMO \cite{zhang2019cell}. Despite the fact that deploying a large number of antennas at base stations can potentially enhance the multiplexing gains and the network capacity, scaling the number of antennas at users leads to several fundamental challenges related to power consumption, implementation complexity, signal processing, channel estimation, spatial separation, and etc.

One promising approach to tackle the aforesaid issues and also achieve the desired goals in the 6G wireless technology is to adopt the fluid antenna system (FAS) as opposed to the traditional antenna system (TAS). FAS refers to any software-manageable fluidic, conductive, or dielectric structure capable of altering its shape and position to reconfigure the radiation characteristics \cite{wong2020fluid}. This is possible due to the recent advances in using flexible conductive substances like liquid metals or ionized solutions \cite{huang2021liquid} and reconfigurable radio-frequency (RF) pixel antennas \cite{song2013efficient}. The latter composes of a massive number of pixels connected by RF switches and by optimizing the RF connections between the pixels gives the effects of changing the position, polarization, frequency and more of radiation in a given space. Implementation and experimental results on FAS have recently been reported in \cite{Shen-tap_submit2024,Zhang-pFAS2024}. It is interesting to note that the pixel-based FAS design in \cite{Zhang-pFAS2024} is particularly appealing as it achieves near-zero-delay position switching for FAS.

In particular, the distinctive attribute of fluid antenna compared with a traditional fixed-position antenna, is the ability to switch positions (also referred to as ports) over a pre-defined space, which proves advantageous in mobile devices due to the physical constraints associated with antenna deployment \cite{wong2022bruce} although FAS can also be as effective at the base station if the channel sees sufficient variations in space. Nevertheless, the fine resolution of the FAS ports means that the channels between the ports can be strongly correlated, thereby making the modeling and analysis of fluid antenna-aided communications much more challenging, see e.g., \cite{wong2022closed,khammassi2023new,ghadi2023copula,ghadi2023gaussian}. 

In designing the 6G mobile communication networks, the unlicensed and shared nature of the wireless spectrum in mobile communication networks causes concerns about reliability and unauthorized access. Hence, it is imperative to address the challenges for secure and ultra-reliable transmission in 6G. Physical layer security (PLS) is a possible approach to enhance the confidentiality and integrity of the transmitted information over wireless channels \cite{ma2024covert}. In the traditional PLS structure, beamforming is the main scheme for boosting communication security, which is achieved by strengthening (weakening) the desired (undesired) signal quality \cite{chen2016survey}. Current beamforming methods are evidently based on a fixed-position antenna array \cite{xiao2023array}. We are thus motivated to consider the potential security advantage of using FAS, which is the goal of this paper.

\subsection{Related Works}
In recent years, many contributions have been made to evaluate the achievable performance of FAS in various communication systems. For instance, to improve the channel model of FAS, \cite{wong2022closed} proposed a simple common correlation parameter to characterize the structure of dependency between the antenna ports. Furthermore, \cite{khammassi2023new} presented an eigenvalue-based model to approximate the spatial correlation given by Jakes' model. Then, \cite{ghadi2023copula} proposed a more general copula-based approach to describe the spatial correlation in FAS, which works under any arbitrary correlated fading distribution. By considering the Gaussian copula, \cite{ghadi2023gaussian} also presented a tight approximation of Jakes' model and revealed that the respective technique can reduce the complexity of analysis and provide mathematical tractability by exploiting rank correlation coefficients.

By employing a single fluid antenna to each mobile user, the authors in \cite{wong2021fluid} proposed the fluid antenna multiple access (FAMA) system, and they derived analytical expressions of the outage probability upper bound and average outage rate lower bound. In order to gain massive connectivity, the authors in \cite{wong2022fast} studied the fast FAMA, where they proposed an approach that can estimate the best port of fluid antenna at every symbol instance without requiring precoding at the base station. However, switching to the best port on a symbol-by-symbol basis by each user is infeasible in practical environments. Motivated by this, the authors in \cite{wong2023slow} proposed slow FAMA, meaning that the fluid antenna of each user updates its best port only if the fading channel changes. Under such an assumption, the authors studied interference immunity by deriving the outage probability upper bound. Moreover, by considering the fluid antenna at a common receiver in the uplink multiuser dirty multiple access channel (MAC), the authors in \cite{ghadi2023fluid} derived the exact closed-form expression of the outage probability under correlated composite fading channels. Their results indicated that the FAS can support a large number of users using only one fluid antenna at the common receiver in a few wavelengths of space. Furthermore, as a special case, the authors in \cite{xu2023outage} analyzed the performance of outage probability for a two-user MAC, in which each user is equipped with the FAS. Besides, by extending the two-user MAC to the multiuser millimeter-wave (mmWave) fluid antenna-assisted scenario, the channel estimation problem and capacity maximization fo FAS were addressed in \cite{xu2023channel,xu2023capacity}. Additionally, the optimization problem of the energy efficiency in FAMA was studied in \cite{chen2023energy,xu2023energy}. Most recently, a new approach called compact ultra massive antenna array (CUMA) utilizing FAS was proposed in \cite{wong2023compact} which is shown to support a large number of users per channel use. Furthermore, by considering a FAS-equipped base station in a multi-user uplink communication system, an optimization problem to minimize the total transmit power of all single-antenna users was formulated and tackled in \cite{Hu2024fluid}.

Additionally, the outage probability analysis for the FAS-aided point-to-point communication system under correlated Nakagami-$m$ fading channels was investigated in \cite{tlebaldiyeva2022enhancing,vega2023simple}. New approximations of diversity gain and outage probability were given in \cite{new2023fluid}. The performance of FAS-assisted backscatter communication systems in terms of the outage probability and delay outage rate was also analyzed in \cite{ghadi2024performance}. To overcome the challenge related to the optimal port selection in FAS, \cite{chai2022port} proposed a machine learning based approach, and their results demonstrated that with only $10\%$ of the ports observed, FAS could perform nearly optimally. Several contributions have recently been made for secure communications utilizing FAS. For example, by assuming that only the legitimate receiver has FAS and other nodes use fixed-position antennas, the average secrecy rate optimization problem was addressed in \cite{tang2023fluid}. In addition, under the scenario in which only the transmitter has movable antennas and the legitimate receiver as well as the eavesdropper consist of a single fixed-position antenna, \cite{cheng2023enabling} investigated the problems of power consumption minimization and secrecy rate maximization. With the same assumption in \cite{cheng2023enabling}, the achievable secrecy rate maximization problem subject to the transmit power budget and the position of all transmit antennas was investigated in \cite{hu2023secure}. The aforesaid studies were mostly based on one-dimensional (1D) FAS. The more general case of two-dimensional (2D) FAS was recently studied in \cite{new2023information} where both sides adopted multiple fluid antennas. It was revealed from the information-theoretic viewpoint that MIMO-FAS outperforms greatly the traditional fixed-position MIMO counterpart in terms of the achievable diversity. ost recently, the integration of FAS with reconfigurable intelligent surface (RIS) was studied in \cite{ghadi2024ris}, where the authors derived analytical expressions for the outage probability and delay outage rate.

It is worth mentioning that a similar notion under the name `movable antenna (MA)' has recently emerged. Conceptually, nevertheless, it is identical to FAS if only position-flexibility is considered, except that MA systems seem to emphasize the use of stepper motors for implementation. A more elaborate commentary between FAS and MA is provided in \cite{zhu2024historical}. There have been interesting results in MA systems. For instance, in \cite{zhu2023modeling}, the authors modeled the general multipath channel for the MA architecture by leveraging the amplitude, phase, and angle of arrival/angle of departure information for each of the multiple channel paths under the far-field condition. The capacity of a MA-enabled point-to-point MIMO system was characterized in \cite{ma2023mimo} by jointly optimizing the positions of the transmit and receive MAs as well as the covariance of transmit signals. Additionally, to reduce the total transmit power and enhance signal strength, the authors in \cite{wang2024movable} formulated a joint optimization problem for MA positions and transmit beamforming in a multiple-input single-output (MISO) interference channel. Furthermore, a MA-enhanced multicast transmission was investigated in \cite{gao2024joint}, where the authors aimed to maximize the minimum weighted signal-to-interference-plus-noise ratio (SINR) among all users by jointly optimizing the position of each transmit/receive MA and the transmit beamforming.

\subsection{Motivations and Contributions}
Motivated by the above discussion, it is important to investigate \textit{(i) How does considering the FAS in the PLS structure improve the reliability and security of wireless communication systems?} and \textit{(ii) how much information can be reliably sent using FAS, so that the eavesdropper cannot decode the confidential message?} In particular, it is useful to understand the benefits of FAS over fixed-position TAS in the PLS setup. Motivated by the above questions, in this paper, we consider the classic wiretap channel, where both the legitimate receiver (Bob) and eavesdropper (Eve) are equipped with a 2D FAS, and our aim is to analyze the secrecy performance.

Following the copula-based approach provided in \cite{ghadi2023gaussian}, we first characterize the distributions of the equivalent channel at the FAS considering the inherent spatial correlation between the fluid antenna ports. Then, we derive compact analytical expressions for the average secrecy capacity (ASC), secrecy outage probability (SOP), and secrecy energy efficiency (SEE) by exploiting the Gauss-Laguerre quadrature technique, where the derived analytical results are valid for any arbitrary choice of correlated fading distribution. Furthermore, we evaluate the system performance in terms of the obtained secrecy metrics under the key benchmark scenario in the literature, i.e., correlated Rayleigh fading. In particular, the main contributions of our work are summarized as follows:
\begin{itemize}
\item We provide general analytical expressions for the joint cumulative distribution function (CDF) and joint probability density function (PDF) of the equivalent channel at Bob and Eve for the FAS-aided secure system.
\item Using the Gauss-Laguerre quadrature approximation, we derive general compact analytical expressions for the ASC, SOP, and SEE for the considered system, which are valid for any choice of arbitrary fading distribution.
\item Additionally, for exemplary purposes, we obtain analytical expressions for the joint CDF and PDF, ASC, SOP, and SEE under correlated Rayleigh fading channels, exploiting the Gaussian copula.
\item We evaluate the security, reliability, and efficiency of the considered FAS-aided secure communications in terms of the ASC, SOP, and SEE. In this regard, we consider the fixed-position single-input multiple-output (SIMO) antenna system that takes advantage of the maximal ratio combining (MRC) technique and antenna selection (AS) under selection combining (SC) scheme as  benchmarks. The numerical results reveal that deploying FAS instead of TAS can significantly enhance the performance of secure transmission in terms of the derived secrecy metrics.
\end{itemize}

\subsection{Paper Organization}
The rest of this paper is organized as follows. Section \ref{sec-sys} describes the system model and presents the secrecy performance analysis of the considered FAS-aided PLS communication system. The statistical characterization is outlined in Section \ref{sub-stat}, while the secrecy metrics under correlated arbitrary and Rayleigh fading are given in Section  \ref{sub-gen} and Section \ref{sub-ray}, respectively. Section \ref{sec-num} illustrates the numerical results, and finally, the conclusions are provided in Section \ref{sec-con}.

\subsection{Notations}
We use boldface upper and lower case letters for matrices and vectors, e.g. $\mathbf{X}$ and $\mathbf{x}$, respectively. $\mathrm{Cov}[\cdot]$ and $\mathrm{Var}[\cdot]$ denote the covariance and variance operators, respectively. Moreover, $(\cdot)^T$, $(\cdot)^{-1}$, $|.|$, and $\textrm{det}(\cdot)$ stand for the transpose, inverse, magnitude, and determinant operators, respectively.

\section{System Model}\label{sec-sys}
We consider a FAS-aided communication setup as shown in Fig.~\ref{fig-system}, where a transmitter (Alice) aims to send a confidential message $x$ with transmit power $P$ to a legitimate receiver (Bob), whereas an eavesdropper (Eve) attempts to decode the information from its received signal. To facilitate the notations, we define subscript $i$ to represent the node associated with Bob and Eve, i.e., $i\in\left\{\mathrm{B}, \mathrm{E}\right\}$. Without loss of generality, we assume that Alice is a single fixed-position antenna transmitter, while Bob and Eve are equipped with a single planar FAS that includes $K_i$ preset positions (i.e., ports), which are distributed evenly over an area of $W_{1,i}\lambda\times W_{2,i}\lambda$ where $\lambda$ defines the wavelength of the carrier frequency. More specifically, a grid structure is considered for the FAS so that $K_{l,i}$ ports are uniformly distributed along a linear space of length $W_{l,i}\lambda$ for $l\in\{1,2\}$, i.e., $K_i=K_{1,i}\times K_{2,i}$ and $W_i=\lambda^2\left(W_{1,i}\times W_{2,i}\right)$. Moreover, we also suppose an applicable mapping function as $\mathcal{F}\left(k_i\right)=\left(k_{1,i},k_{2,i}\right)$ and $\mathcal{F}^{-1}\left(k_{1,i},k_{2,i}\right)=k_i$ to convert the 2D indices to the 1D index, in which $k_i\in\{1,\dots, K_i\}$ and $k_{l,i}\in\{1,\dots, K_{l,i}\}$. Besides, it is assumed that each node $i$ can only switch its FAS to one activated port. Therefore, the received signal at the node $i$ can be expressed as
%
%
%
%
%
%
%
%
%
%
%
%
%
%
\begin{align}
	\mathbf{y}_i=\mathbf{h}_ix+\mathbf{z}_i,
\end{align}
where $\mathbf{z}_i$ denotes the independent identically distributed (i.i.d.) additive white Gaussian noise (AWGN) vector with zero mean and variance $\delta^2_i$ at each port of node $i$. Besides, $\mathbf{h}_{i}$ denotes the complex channel from Alice to the ports at Bob and Eve, which can be modeled as \cite{new2023fluidnew}
\begin{align}
\mathbf{h}_i=&\,\sqrt{\frac{\mathcal{N}}{\mathcal{N}+1}}\mathrm{e}^{j\omega}\mathbf{a}\left(\theta_{0,i},\psi_{0,i}\right)\notag\\
&+\sqrt{\frac{1}{\mathcal{M}\left(\mathcal{N}+1\right)}}\sum_{m=1}^\mathcal{M}\zeta_{m,i}\mathbf{a}\left(\theta_{m,i},\psi_{m,i}\right), \label{eq-h}
\end{align}
in which $\mathcal{N}$ is the Rice factor, $\omega$ denotes the phase of the line-of-sight (LoS) component, $\zeta_{m,i}$ defines the complex channel coefficient of the $m$-th scattered component, and $\mathcal{M}$ is the number of non-LoS paths. Moreover, $\mathbf{a}\left(\theta_{m,i},\psi_{m,i}\right)$ is the receive steering vector, which is defined as \cite{new2023fluidnew}
\begin{align}\notag
&\mathbf{a}\left(\theta_{m,i},\psi_{m,i}\right)\\
&=\left[1\quad\mathrm{e}^{j\frac{2\pi W_{1,i}}{K_{1,i}-1}\sin\theta_{m,i}\cos\psi_{m,i}}\dots~ \mathrm{e}^{j2\pi W_{1,i}\sin\theta_{m,i}\cos\psi_{m,i}}\right]^T\notag\\
&\otimes\left[1\quad\mathrm{e}^{j\frac{2\pi W_{2,i}}{K_{2,i}-1}\sin\theta_{m,i}\cos\psi_{m,i}}\dots~ \mathrm{e}^{j2\pi W_{2,i}\sin\theta_{m,i}\cos\psi_{m,i}}\right],
\end{align}
where $\theta_{m,i}$ and $\psi_{m,i}$ denote the azimuth and elevation angle-of-arrival, respectively, and $\otimes$ is the Kronecker tensor product.

\begin{figure*}[!t]
\centering
\includegraphics[width=.8\linewidth]{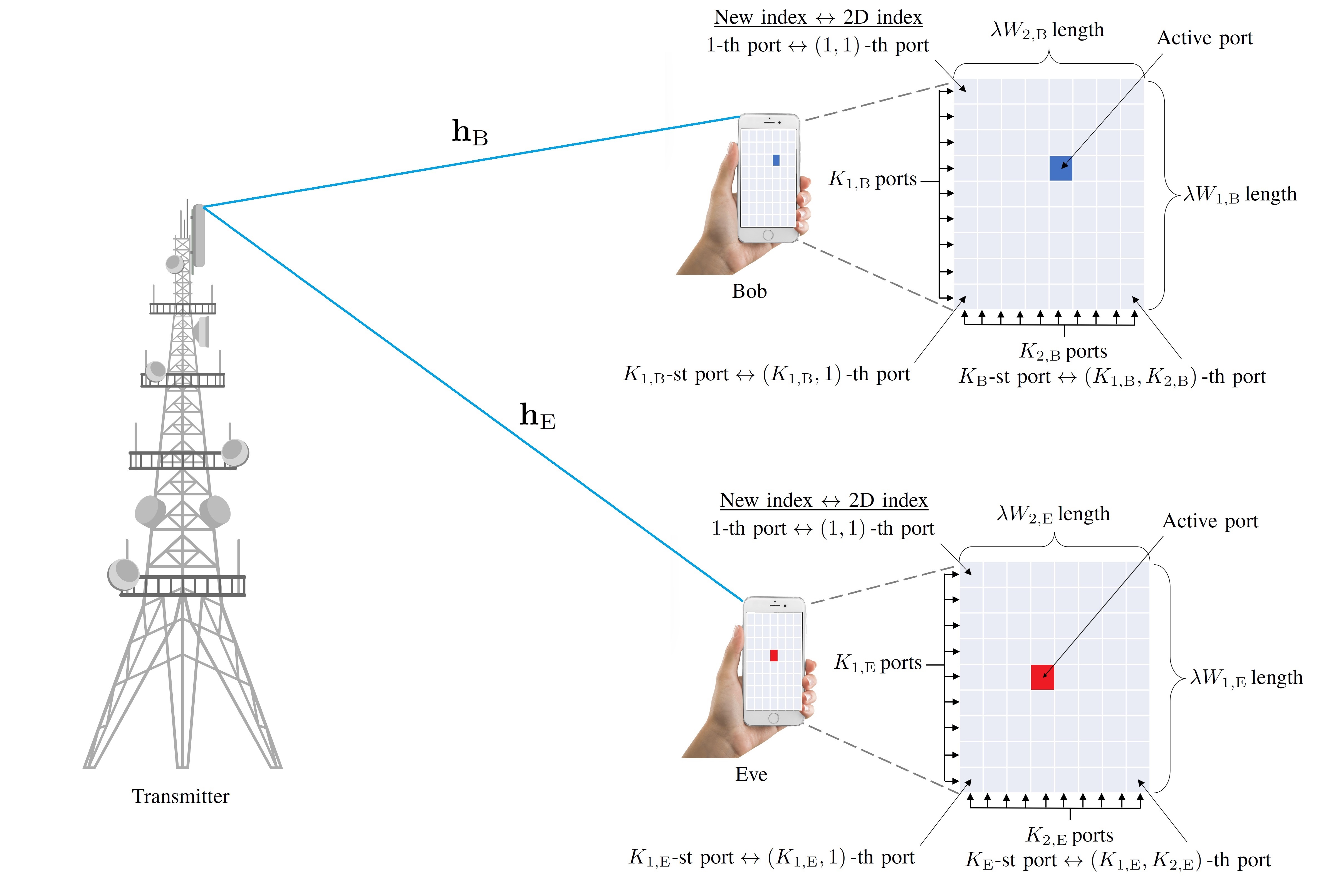}
\caption{The system model of a secure FAS communication system.}\label{fig-system}
\end{figure*}

Note that the fluid antenna ports can be arbitrarily close to each other and hence, the corresponding channels are spatially correlated. Assuming a three-dimensional (3D) isotropic scattering environment, the covariance between two arbitrary ports $k_i=\mathcal{F}^{-1}\left(k_{1,i},k_{2,i}\right)$ and $\tilde{k}_i=\mathcal{F}^{-1}(\tilde{k}_{1,i},\tilde{k}_{2,i})$ for the 2D fluid antenna at node $i$ is given by \cite{new2023information}
\begin{align}\label{eq-cov-gen}
\varrho_{k_i,\tilde{k}_i}\triangleq\mathrm{cov}[k_i,\tilde{k}_i]=\mathcal{J}_0\left(2\pi\Vert k_i-\tilde{k}_i\Vert\right),
\end{align}
in which is $\mathcal{J}_0\left(t\right)=\frac{\sin t}{t}$ defines the zero-order spherical Bessel function of the first kind.

Further, we assume that only the best port that maximizes the received signal-to-noise ratio (SNR) at node $i$ is activated.\footnote{We consider the case that the channel state information (CSI) is available at Eve, and therefore Eve has the ability to select the best port to maximize the SNR. This can be viewed as the worst case as Eve is given the best condition to operate. In addition, given the inherent spatial correlation between fluid antenna ports, acquiring the full CSI requires only a small number of port observations, regardless of the number of fluid antenna ports.} As a result, the received instantaneous SNR at node $i$ for the considered FAS can be expressed as
\begin{align}
\gamma_i=\frac{P\left|\left[\mathbf{h}_i\right]_{k_i^*}\right|^2}{d_i^\nu\delta_i^2}=\bar{\gamma}_i\left|\left[\mathbf{h}_i\right]_{k_i^*}\right|^2,
\label{eq:SNR}
\end{align} 
where $d_i$ denotes the distance between the base station and node $i$, $\nu$ is the path-loss exponent, and $\bar{\gamma}_i=\frac{P}{d_i^\nu\delta^2_i}$ can be interpreted as the average SNR at node $i$ for a $1$-element fluid antenna. For the sake of convenience \cite{jeon2011bounds}, we define a SNR degradation factor $\kappa=\bar\gamma_{\rm B}/\bar\gamma_{\rm E}$ to account for the scenarios on which the legitimate user has an SNR advantage\footnote{Note that such SNR advantage may come from either the difference in path-losses, or a distinct noise power spectral density (for instance, due to dissimilar noise figures at the receiver ends).} over Eve ($\kappa>1$), disadvantage ($\kappa<1$) or balance ($\kappa=1$).
The term $k_{i^*}$ denotes the index of the best port at node $i$, i.e., 
\begin{align}
k_i^*=\arg\underset{k_i}{\max}\left\{\left|\left[\mathbf{h}_i\right]_{k_i}\right|^2\right\}.
\end{align}
Therefore, the channel gain at node $i$ for the considered FAS can be given by
\begin{align}
g_{\mathrm{FAS},i}=\max\left\{g_{1_i},\dots,g_{K_i}\right\},\label{eq-g-fa}
\end{align}
in which $g_{k_i}=\left|h_{k_i}\right|^2$ is the channel gain of port $k$ at node $i$. 

\section{Secrecy Performance Analysis}
Here, we first characterize the distribution of the SNR at Bob and Eve, exploiting the copula approach. Second, we derive analytical expressions for the ASC, SOP, and SEE under arbitrary correlated fading channels with the help of the Gauss-Laguerre quadrature. Afterwards, in order to evaluate the system performance, we provide secrecy metrics under correlated Rayleigh fading distributions for exemplary purposes. 

\subsection{Statistical Characterization}\label{sub-stat}
Given that only the best port at node $i$ with the maximum signal envelope is considered for communication, it is necessary to find the distribution of the channel gain $g_{\mathrm{FA}_i}$ provided in \eqref{eq-g-fa}. For this purpose, we exploit Sklar's theorem, which is based on copula theory. In particular, a $d$-dimension copula $C: \left[0,1\right]^d\rightarrow \left[0,1\right]$ is a joint CDF of $d$ random vectors on the unit cube $\left[0,1\right]^d$ with uniform marginal distributions, i.e., \cite{nelsen2006introduction}
\begin{align}\label{eq-cop-def}
C\left(u_1,\dots, u_d;\vartheta_C\right)=\Pr\left(U_1\leq u_1,\dots, U_d\leq u_d\right),
\end{align}
in which $u_q=F_{S_q}\left(s_q\right)$ and $S_q$ is an arbitrary random variable for $q\in\left\{1,\dots, d\right\}$ and $\vartheta_C$ denotes the copula parameter, which measures the structure of dependency between arbitrary correlated random variables. In particular, the copula parameter $\vartheta_C$ has a relation with the Pearson correlation coefficient and popular rank correlation coefficients such as the Spearman's $\rho$ and the Kendall's $\tau$. Hence, the covariance between two arbitrary ports provided in \eqref{eq-cov-gen} can be expressed in terms of copulas. As such, the spatial correlation between fluid antenna ports can be flexibly considered in SNR's distributions. Moreover, the importance of the copula technique is due to Sklar's theorem, which states that for an arbitrary $d$-dimension CDF $F_{S_1,\dots, S_d}$$\left(s_1,\dots, s_d\right)$ with univariate marginal distributions $F_{S_q}\left(s_q\right)$, there exists a copula function $C$ such that for all $s_q$ in the extended real line domain
$\mathbb{R}$ \cite{nelsen2006introduction}, 
\begin{align}
F_{S_1,\dots, S_d}\left(s_1, \dots, s_d\right)=C\left(F_{S_1}\left(s_1\right),\dots, F_{S_d}\left(s_d\right);\vartheta_C\right). \label{eq-sklar}
\end{align}
\begin{theorem}\label{thm-dist-gen}
The CDF and PDF of $g_{\mathrm{FA},i}$ at node $i$ for the FAS-aided secure communication system under arbitrary correlated fading channel channels are, respectively, given by
\begin{align}
F_{g_{\mathrm{FAS},i}}(r)=C\left(F_{g_{1_i}}\left(r\right),\dots, F_{g_{K_i}}\left(r\right);\vartheta_C\right), \label{eq-cdf-gen}
\end{align}
and
\begin{align}
	&\hspace{-3mm}f_{g_{\mathrm{FAS},i}}(r)=\prod_{\hat{k}_i=1}^{k_i}f_{g_{\hat{k}_i}}\left(r\right) c\left(F_{g_{1_i}}\left(r\right),\dots, F_{g_{K_i}}\left(r\right);\vartheta_C\right),\label{eq-pdf-gen}
\end{align}
in which $F_{g_{1_i}}\left(r\right)$ and $f_{g_{1_i}}\left(r\right)$ are the marginal CDF and PDF of the channel gain $g_{k_i}$ with an arbitrary distribution. Besides, $c\left(.\right)$ is the copula density function of an arbitrary copula $C(.)$ that is given by \cite{nelsen2006introduction}
\begin{align}\notag
&c\left(F_{g_{1_i}}\left(r\right),\dots, F_{g_{K_i}}\left(r\right);\vartheta_C\right)\\
&=\frac{\partial^{K_i}C\left(F_{g_{1_i}}\left(r\right),\dots, F_{g_{K_i}}\left(r\right);\vartheta_C\right)}{\partial F_{g_{1_i}}\left(r\right)\dots \partial F_{g_{K_i}}\left(r\right)}.
\end{align}

\begin{proof}
By extending the CDF definition, we have
\begin{subequations}
\begin{align}
F_{g_{\mathrm{FAS},i}}\left(r\right)&=\Pr\left(\max\left\{g_{1_i},\dots,g_{K_i}\right\}\leq r\right)\\
&=\Pr\left(g_{1_i}\leq r,\dots, g_{K_i}\leq r\right)\\
&=F_{g_{1_i},\dots, g_{K_i}}\left(r,\dots, r\right)\\
&\hspace{-.5mm}\overset{(a)}{=}C\left(F_{g_{1_i}}\left(r\right),\dots, F_{g_{K_i}}\left(r\right);\vartheta_C\right),
\end{align}
\end{subequations}
where $(a)$ is derived by using Sklar's theorem in \eqref{eq-sklar}. Moreover, by applying the chain rule to \eqref{eq-cdf-gen}, the PDF $f_{g_{\mathrm{FAS},i}}\left(r\right)$ can be derived as \eqref{eq-pdf-gen} and the proof is completed. 
\end{proof}
\end{theorem}

\subsection{Secrecy Metrics under Arbitrary Fading Channels}\label{sub-gen}
\subsubsection{ASC Analysis}
The secrecy capacity of the FAS system is defined as the difference between the channel capacities corresponding to legitimate and wiretap links, i.e.,
\begin{align}
\hspace{-2mm}\mathcal{C}_\mathrm{s}\left(\gamma_\mathrm{B},\gamma_\mathrm{E}\right)=\max\left\{\log_2\left(1+\gamma_\mathrm{B}\right)-\log_2\left(1+\gamma_\mathrm{E}\right),0\right\}.\label{eq-def-sc}
\end{align}
Noting that $\mathcal{C}_\mathrm{s}\left(\gamma_\mathrm{B},\gamma_\mathrm{E}\right)$ is a random function of $\gamma_{i}$ as it depends on the channel, the corresponding ASC can be defined as
\begin{align}
	&\bar{\mathcal{C}}_\mathrm{s}\left[\mathrm{bps/Hz}\right]=\mathbb{E}\left[\mathcal{C}\left(\gamma_{{\mathrm{B}}},\gamma_{{\mathrm{E}}}\right)\right]\label{eq-def1-asc}\\ 
	&=\int_{0}^\infty\hspace{-2mm}\int_0^{\gamma_{\mathrm{B}}} \mathcal{C}_\mathrm{s}\left(\gamma_\mathrm{B},\gamma_\mathrm{E}\right)f_{\gamma_\mathrm{B}}\left(\gamma_\mathrm{B}\right)f_{\gamma_\mathrm{E}}\left(\gamma_\mathrm{E}\right)d\gamma_\mathrm{E}d\gamma_\mathrm{B}.\label{eq-def-asc}
\end{align}

\begin{Proposition}\label{pr-asc}
The ASC for the FAS under arbitrary fading channels is given by
\begin{align}\nonumber
	\bar{\mathcal{C}}_\mathrm{s}\hspace{-1mm}\approx&\hspace{-1mm} \sum_{\tilde{n}=1}^{\tilde{N}}\frac{w_{\tilde{n}}\mathrm{e}^{\epsilon_{\tilde{n}}}}{1+\epsilon_{\tilde{n}}}\hspace{-1mm}\left[1-C\left(F_{g_{1_\mathrm{B}}}\left(\frac{\epsilon_{\tilde{n}}}{\bar{\gamma}_\mathrm{B}}\right),\dots, F_{g_{K_\mathrm{B}}}\left(\frac{\epsilon_{\tilde{n}}}{\bar{\gamma}_\mathrm{B}}\right);\vartheta_C\right)\right]\\
	&\times C\left(F_{g_{1_\mathrm{E}}}\left(\frac{\epsilon_{\tilde{n}}}{\bar{\gamma}_\mathrm{E}}\right),\dots, F_{g_{K_\mathrm{E}}}\left(\frac{\epsilon_{\tilde{n}}}{\bar{\gamma}_\mathrm{E}}\right);\vartheta_C\right),\label{eq-asc-gen}
\end{align}
in which 
\begin{align}\label{eq-w} w_{\tilde{n}}=\frac{\epsilon_{\tilde{n}}}{2\left(\tilde{N}+1\right)^2L^2_{\tilde{N}+1}\left(\epsilon_{\tilde{n}}\right)},
\end{align}
$\epsilon_{\tilde{n}}$ is the $\tilde{n}$-th root of Laguerre polynomial $L_{\tilde{N}}\left(\epsilon_{\tilde{n}}\right)$, $\tilde{N}$ defines the parameter to ensure a complexity-accuracy trade-off \cite{abromowitz1972handbook}.
\end{Proposition}

\begin{proof}
See Appendix \ref{app-asc}.
\end{proof}

\subsubsection{SOP Analysis}
SOP is defined as the probability that the instantaneous secrecy capacity $\mathcal{C}_\mathrm{s}$ is less than a target secrecy rate $R_\mathrm{s}\ge0$, i.e.,
\begin{align}
	P_\mathrm{sop}=\Pr\left(\mathcal{C}_\mathrm{s}\leq R_\mathrm{s}\right).\label{eq-def-sop}
\end{align}

\begin{Proposition}\label{pro-sop}
The SOP for the FAS under arbitrary fading channels is given by  
\begin{align}\notag
	&P_{\mathrm{sop}}\approx \sum_{\tilde{n}=1}^{\tilde{N}}\frac{w_{\tilde{n}}\mathrm{e}^{\epsilon_{\tilde{n}}}}{1+\epsilon_{\tilde{n}}} \prod_{k_\mathrm{E}=1}^{K_\mathrm{E}} \frac{f_{g_{k_\mathrm{E}}}\left(\frac{\epsilon_{\tilde{n}}}{\bar{\gamma}_\mathrm{E}}\right)}{\bar{\gamma}_\mathrm{E}^{k_\mathrm{E}}}\\ \notag
	&\times c\left(F_{g_{1_\mathrm{E}}}\left(\frac{\epsilon_{\tilde{n}}}{\bar{\gamma}_\mathrm{E}}\right),\dots, F_{g_{K_\mathrm{E}}}\left(\frac{\epsilon_{\tilde{n}}}{\bar{\gamma}_\mathrm{E}}\right);\vartheta_C\right)
	\\
	&\times C\left(F_{g_{1_\mathrm{B}}}\left(\frac{R_\mathrm{o}\epsilon_{\tilde{n}}+R_\mathrm{t}}{\bar{\gamma}_\mathrm{B}}\right),\dots, F_{g_{K_\mathrm{B}}}\left(\frac{R_\mathrm{o}\epsilon_{\tilde{n}}+R_\mathrm{t}}{\bar{\gamma}_\mathrm{B}}\right);\vartheta_C\right), \label{eq-sop-gen}
\end{align}
where  $R_\mathrm{o}=2^{R_{\mathrm{s}}}$, $R_\mathrm{t}=R_\mathrm{o}-1$, and $w_{\tilde{n}}$ is defined in \eqref{eq-w}.
\end{Proposition}

\begin{proof}
See Appendix \ref{app-sop}.
\end{proof}

\subsubsection{SEE Analysis} SEE stands as a momentous secrecy metric for evaluating the efficacy of wireless communication systems within energy consumption, which is denoted as the ratio of ASC to the total power consumption. In fact, SEE is a performance metric that indicates how a communication system with security considerations efficiently exploits energy to guarantee a secure transmission \cite{ghadi2023newperformance}. In general, although increasing the transmit power can provide a higher throughput at the receiver nodes, it can also increase energy consumption and interference. Hence, to capture such a trade-off, the SEE for the considered FAS can be defined as
\begin{align}
\mathcal{E}_\mathrm{s}\left[\mathrm{bps/Hz/J}\right]=\frac{\bar{\mathcal{C}_\mathrm{s}}}{P_\mathrm{tot}}=\frac{\bar{\mathcal{C}_\mathrm{s}}}{P/\alpha+P_\mathrm{c}+\breve{K}P_\mathrm{act}},\label{eq-def-see}
\end{align} 
in which $\alpha$ is the drain efficiency of the high-power amplifier (HPA), $P_\mathrm{c}$ denotes the constant circuit power at Bob, $\breve{K}$ defines the number of activated ports, and $P_\mathrm{act}$ represents the required power for activating a given port.\footnote{It is reasonable to recognize that the power consumption in the FAS is less than that in the MRC SIMO system. This is because in the FAS, only one port with the maximum SNR is activated whereas in the MRC SIMO system, all antennas are activated, which consumes more power.}
 
\begin{Proposition}
The SEE for the FAS under arbitrary fading channels is given by  
\begin{align}\nonumber
\mathcal{E}_\mathrm{s}\hspace{-1mm}=&\hspace{-1mm} \sum_{\tilde{n}=1}^{\tilde{N}}\frac{w_{\tilde{n}}\mathrm{e}^{\epsilon_{\tilde{n}}}}{1+\epsilon_{\tilde{n}}}\hspace{-1mm}\left[1-C\left(F_{g_{1_\mathrm{B}}}\left(\frac{\epsilon_{\tilde{n}}}{\bar{\gamma}_\mathrm{B}}\right),\dots, F_{g_{K_\mathrm{B}}}\left(\frac{\epsilon_{\tilde{n}}}{\bar{\gamma}_\mathrm{B}}\right);\vartheta_C\right)\right]\\
		&\times \frac{C\left(F_{g_{1_\mathrm{E}}}\left(\frac{\epsilon_{\tilde{n}}}{\bar{\gamma}_\mathrm{E}}\right),\dots, F_{g_{K_\mathrm{E}}}\left(\frac{\epsilon_{\tilde{n}}}{\bar{\gamma}_\mathrm{E}}\right);\vartheta_C\right)}{P/\alpha+P_\mathrm{c}+P_\mathrm{act}},\label{eq-see-gen}
\end{align}
where $w_{\tilde{n}}$ is defined in \eqref{eq-w}.
\end{Proposition}
\begin{proof}
The proof is done by substituting \eqref{eq-asc-gen} into \eqref{eq-def-see}.
\end{proof}

\begin{remark}
The derived analytical secrecy metrics in \eqref{eq-asc-gen}, \eqref{eq-sop-gen}, and \eqref{eq-see-gen} are valid for any arbitrary choice of fading distributions, which can also describe unknown dependence structures. Hence, by only considering the marginal distribution of respective fading channels, the ASC, SOP, and SEE for the FAS can be directly obtained. 
\end{remark}

\begin{remark}
The parameter $\vartheta_C$ is defined in terms of the fluid antenna size $W_i$ and the number of fluid antenna ports $K_i$, which can measure the spatial correlation between fluid antenna ports based on \eqref{eq-cov-gen}. Note that $\vartheta_C$ generally depends on the choice of copula, and it can be expressed in terms of the rank correlation coefficients.
\end{remark}

\begin{remark}
In general, any function that has the following properties can be considered as a copula $C$ for the obtained distributions in \eqref{eq-cdf-gen} and \eqref{eq-pdf-gen}, and the derived secrecy metrics in \eqref{eq-asc-gen}, \eqref{eq-sop-gen}, \eqref{eq-see-gen}:
1) For every $q$-dimensional vector $\mathbf{u}$ in $[0,1]^d$, $C(\mathbf{u})=0$ if at least one coordinate of $\mathbf{u}$ is $0$, and if all coordinates of $\mathbf{u}$ are $1$ except $u_{q'}$, then $C(\mathbf{u})=u_{q'}$; 2) For every $\mathbf{a}$ and $\mathbf{b}$ in $[0,1]^q$ such that $\mathbf{a}\leq\mathbf{b}$, we have $V_C\left([\mathbf{a},\mathbf{b}]\right)\ge 0$, where $V_C(B)$ is the C-volume of $B$ \cite{nelsen2006introduction}. However, since each $C$ has its own mathematical and statistical properties, choosing an appropriate $C$ should be based on the purpose of the study. As shown in \cite{ghadi2023gaussian}, the Gaussian copula can be considered as a tight approximation for the spatial correlation in the FAS.
\end{remark}

\subsection{Secrecy Metrics under Rayleigh Fading Channels}\label{sub-ray}
We now specialize our general framework to a particular case of relevance, which corresponds to the case of correlated Rayleigh fading. In this case, the marginal PDF and CDF with the scale parameter $\eta_i$ are respectively given by
\begin{align}
F_{g_{k_i}}\left(r\right)&=1-\mathrm{e}^{-\eta_ir},\label{eq-marg-p}\\
f_{g_{k_i}}\left(r\right)&=\eta_i\mathrm{e}^{-\eta_ir}. \label{eq-marg-c}
\end{align}
This corresponds to the case that $\mathcal{N}=0$ and $\mathcal{M}\rightarrow\infty$ in \eqref{eq-h}. Consequently, the amplitude of each entry of $\mathbf{h}_i$ follows the Rayleigh distribution. Under such an assumption, the corresponding covariance between two arbitrary ports $k_i$ and $\tilde{k}_i$ in \eqref{eq-cov-gen} can be written as \cite{new2023information}
\begin{multline}
\varrho_{k_i,\tilde{k}_i}=\Omega_i\times\\
\mathcal{J}_0\left(\hspace{-1mm}2\pi\sqrt{\left(\frac{k_{1,i}-\tilde{k}_{1,i}}{K_{1,i}-1}W_{1,i}\right)^2\hspace{-3mm}+\left(\frac{k_{2,i}-\tilde{k}_{2,i}}{K_{2,i}-1}W_{2,i}\right)^2}\right)\hspace{-1mm}, \label{eq-rho}
\end{multline}
in which $\Omega_i$ denotes the variance of the entry of $\mathbf{h}_i$. Next, by considering a specific copula function $C(.)$ and the corresponding copula density function $c(.)$, the CDF and PDF of $g_{\mathrm{FAS},i}$ can be derived under correlated Rayleigh fading channels. For this purpose, it is necessary to select some copula that can accurately describe the respective spatial correlation. Finding an appropriate copula for high-dimensional random variables is often performed based on empirical statistic data analysis and desired properties of the copula model. When this alternative is not possible due to the lack of such empirical data, there are many popular parametric copulas such as Archimedean families (e.g., Clayton, Gumbel, Frank, and Farlie-Gumbel-Morgenstern (FGM)) and the elliptical copula (including the Gaussian and Student-$t$ copulas) that can be exploited for the performance analysis of various wireless communication systems \cite{ghadi2020copula,gholizadeh2015capacity,ghadi2022capacity,zheng2019copula,ghadi2020copula1,ghadi2022performance,choi2015copula,ghadi2022impact,ghadi2024newperformance}. Since the number of ports in the 2D FAS can be very large, we exploit the Gaussian copula for performance analysis due to its unique properties in capturing both linear and elliptical dependencies in high-dimensional settings, providing mathematical tractability, and including a correlation matrix. The $d$-dimension Gaussian copula with correlation matrix $\mathbf{R}\in\left[-1,1\right]^{d\times d}$ is defined as \cite{ghadi2023gaussian}
\begin{align}
\hspace{-2mm}	C_\mathrm{G}\left(u_1,\dots,u_d;\vartheta_\mathrm{G}\right)=\Phi_\mathbf{R}\left(\varphi^{-1}(u_1),\dots,\varphi^{-1}(u_d);\vartheta_\mathrm{G}\right)\hspace{-1mm},\hspace{-2mm} \label{eq-gc-def1}
\end{align}
in which $\varphi^{-1}(u_q)=\sqrt{2}\mathrm{erf}^{-1}\left(2u_q-1\right)$ denotes the inverse CDF (i.e., quantile function) of the standard normal distribution, where $\mathrm{erf}^{-1}\left(\cdot\right)$ is the inverse of the error function $\mathrm{erf}\left(z\right)=\frac{2}{\sqrt{\pi}}\int_0^z\mathrm{e}^{-t^2}dt$. The term $\Phi_\mathbf{R}(\cdot)$ is the joint CDF of the multivariate normal distribution with zero mean vector and correlation matrix $\mathbf{R}$ and also $\vartheta_\mathrm{G}$ defines the dependence parameter of the Gaussian copula, which controls the degree of dependence between correlated random variables.  Furthermore, the Gaussian copula density function is written as
\begin{align}
c_\mathrm{G}\left(u_1,\dots,u_d;\vartheta_\mathrm{G}\right)=\frac{\exp\left(-\frac{1}{2}\left(\boldsymbol{\varphi}^{-1}\right)^T\left(\mathbf{R}^{-1}-\mathbf{I}\right)\boldsymbol{\varphi}^{-1}\right)}{\sqrt{{\rm det}\left(\mathbf{R}\right)}}, \label{eq-gc-def2}
\end{align}
where $\boldsymbol{\varphi}^{-1}$ is a vector that includes the quantile function of the standard normal distribution $\varphi^{-1}\left(u_q\right)$, $\mathrm{det}\left(\mathbf{R}\right)$ is the determinant of $\mathbf{R}$, and $\mathbf{I}$ is the identity matrix. The Gaussian copula $C_\mathrm{G}$ is constructed using the inversion method based on Sklar's theorem and the Gaussian copula density function $c_\mathrm{G}$ is derived by using the chain rule. The detailed construction process is provided in Appendix \ref{app-g}.

Therefore, by using the Gaussian copula, the corresponding correlation matrix $\mathbf{R}_{k_i,\tilde{k}_i}$ of FAS can be expressed in terms of the covariance matrix with the help of Cholesky decomposition \cite{ghadi2023gaussian}. Hence, the correlation between two arbitrary ports $k_i$ and $\tilde{k}_i$ (i.e., each entry of $\mathbf{R}_{k_i,\tilde{k}_i}$) is given by
\begin{align}
\vartheta_{G_{k_i,\tilde{k}_i}}\triangleq \mathrm{corr}[k_i,\tilde{k}_i]=\frac{\mathrm{cov}[k_i,\tilde{k}_i]}{\sqrt{\mathrm{var}\left[k_i\right]\mathrm{var}[\tilde{k}_i]}}=\frac{\varrho_{k_i,\tilde{k}_i}}{\Omega_i}. \label{eq-g-par}
\end{align}
\begin{corollary}
The CDF and PDF of $g_{\mathrm{FAS},i}$ at node $i$ for the FAS under correlated Rayleigh fading channels are given by \eqref{eq-cdf-ray} and \eqref{eq-pdf-ray} (see top of next page), in which $\vartheta_{G_{k_i,\tilde{k}_i}}$ is defined in \eqref{eq-g-par} and 
\begin{figure*}[t]
\begin{align}
F_{g_{\mathrm{FAS},i}}(r)&= \Phi_{\mathbf{R}_{k_i,\tilde{k}_i}}\left(\sqrt{2}\mathrm{erf}^{-1}\left(1-2\mathrm{e}^{-\eta_ir}\right),\dots,\sqrt{2}\mathrm{erf}^{-1}\left(1-2\mathrm{e}^{-\eta_ir}\right);\vartheta_{G_{k_i,\tilde{k}_i}}\right),\label{eq-cdf-ray}\\
f_{g_{\mathrm{FAS},i}}(r)&={\eta^{K_i}_i}\mathrm{e}^{-K_i\eta_ir} \frac{\exp\left(-\frac{1}{2}\left(\boldsymbol{\varphi}_{g_{k_i}}^{-1}\right)^T\left(\mathbf{R}_{k_i,\tilde{k}_i}^{-1}-\mathbf{I}\right)\boldsymbol{\varphi}_{g_{k_i}}^{-1}\right)}{\sqrt{{\rm det}\left(\mathbf{R}_{k_i,\tilde{k}_i}\right)}},\label{eq-pdf-ray}
\end{align}
\hrulefill
\end{figure*}
\begin{align}
\boldsymbol{\varphi}^{-1}_{g_{k_i}}=\left[\sqrt{2}\mathrm{erf}^{-1}\left(1-2\mathrm{e}^{-\eta_ir}\right),\dots,\sqrt{2}\mathrm{erf}^{-1}\left(1-2\mathrm{e}^{-\eta_ir}\right)\right]^T\hspace{-3mm}.
\end{align}

\begin{proof}
By plugging the marginal distributions from \eqref{eq-marg-p} and \eqref{eq-marg-c} into \eqref{eq-gc-def1} and \eqref{eq-gc-def2}, and then using \eqref{eq-g-par}, the proof is done.
\end{proof}
\end{corollary}

\begin{corollary}
The ASC for the FAS under correlated Rayleigh fading channels is given by \eqref{eq-asc-ray} (see top of this page), in which $\vartheta_{G_{k_i,\tilde{k}_i}}$ is defined in \eqref{eq-g-par}.
\begin{figure*}[t]
\begin{align}\notag
\bar{\mathcal{C}}_\mathrm{s}\approx& \sum_{\tilde{n}=1}^{\tilde{N}}\frac{w_{\tilde{n}}\mathrm{e}^{\epsilon_{\tilde{n}}}}{1+\epsilon_{\tilde{n}}}\left[1-\Phi_{\mathbf{R}_{k_\mathrm{B},\tilde{k}_\mathrm{B}}}\left(\sqrt{2}\mathrm{erf}^{-1}\left(1-2\mathrm{e}^{-\frac{\eta_\mathrm{B}\epsilon_{\tilde{n}}}{\bar{\gamma}_\mathrm{B}}}\right),\dots,\sqrt{2}\mathrm{erf}^{-1}\left(1-2\mathrm{e}^{-\frac{\eta_\mathrm{B}\epsilon_{\tilde{n}}}{\bar{\gamma}_\mathrm{B}}}\right);\vartheta_{G_{k_\mathrm{B},\tilde{k}_\mathrm{B}}}\right)\right]\\
&\times \Phi_{\mathbf{R}_{k_\mathrm{E},\tilde{k}_\mathrm{E}}}\left(\sqrt{2}\mathrm{erf}^{-1}\left(1-2\mathrm{e}^{-\frac{\eta_\mathrm{E}\epsilon_{\tilde{n}}}{\bar{\gamma}_\mathrm{E}}}\right),\dots,\sqrt{2}\mathrm{erf}^{-1}\left(1-2\mathrm{e}^{-\frac{\eta_\mathrm{E}\epsilon_{\tilde{n}}}{\bar{\gamma}_\mathrm{E}}}\right);\vartheta_{G_{k_\mathrm{E},\tilde{k}_\mathrm{E}}}\right)\label{eq-asc-ray}
\end{align}
\hrulefill
\end{figure*}
\end{corollary}

\begin{proof}
By inserting the marginal CDF $F_{g_{k_i}}\left(r\right)$ into \eqref{eq-asc-gen} and the Gaussian copula definition as well as \eqref{eq-g-par}, the proof is accomplished. 
\end{proof}

Besides, note that the asymptotic behaviour of the ASC in the high SNR regime, i.e., $\bar{\gamma}_\mathrm{B}$, is given by \eqref{eq-asc-asy}.
\begin{figure*}[t]
\begin{align}
\bar{\mathcal{C}}_\mathrm{s}^\infty\approx \sum_{\tilde{n}=1}^{\tilde{N}}\frac{w_{\tilde{n}}\mathrm{e}^{\epsilon_{\tilde{n}}}}{1+\epsilon_{\tilde{n}}} \Phi_{\mathbf{R}_{k_\mathrm{E},\tilde{k}_\mathrm{E}}}\left(\sqrt{2}\mathrm{erf}^{-1}\left(1-2\mathrm{e}^{-\frac{\eta_\mathrm{E}\epsilon_{\tilde{n}}}{\bar{\gamma}_\mathrm{E}}}\right),\dots,\sqrt{2}\mathrm{erf}^{-1}\left(1-2\mathrm{e}^{-\frac{\eta_\mathrm{E}\epsilon_{\tilde{n}}}{\bar{\gamma}_\mathrm{E}}}\right);\vartheta_{G_{k_\mathrm{E},\tilde{k}_\mathrm{E}}}\right)\label{eq-asc-asy}
\end{align}
\hrulefill
\end{figure*}

\begin{corollary}
The SOP for the considered 2D fluid antenna-aided secure communication system under correlated Rayleigh fading channels is given by \eqref{eq-sop-ray} (see top of this page), in which $\vartheta_{G_{k_i,\tilde{k}_i}}$ is defined in \eqref{eq-g-par}.
\begin{figure*}[t]
\begin{align}\notag
P_\mathrm{sop}\approx& \sum_{\tilde{n}=1}^{\tilde{N}}\frac{w_{\tilde{n}}\mathrm{e}^{\epsilon_{\tilde{n}}}{\lambda^{K_\mathrm{E}}_\mathrm{E}}\mathrm{e}^{-\frac{K_\mathrm{E}\eta_\mathrm{E}\epsilon_{\tilde{n}}}{\bar{\gamma}_\mathrm{E}}}}{\left(1+\epsilon_{\tilde{n}}\right) \bar{\gamma}_\mathrm{E}^{K_\mathrm{E}}\sqrt{{\rm det}\left(\mathbf{R}_{k_\mathrm{E},\tilde{k}_\mathrm{E}}\right)}} \exp\left(-\frac{1}{2}\left(\boldsymbol{\varphi}_{g_{k_\mathrm{E}}}^{-1}\right)^T\left(\mathbf{R}_{k_\mathrm{E},\tilde{k}_\mathrm{E}}^{-1}-\mathbf{I}\right)\boldsymbol{\varphi}_{g_{k_\mathrm{E}}}^{-1}\right)\\
&\times \Phi_{\mathbf{R}_{k_\mathrm{B},\tilde{k}_\mathrm{B}}}\left(\sqrt{2}\mathrm{erf}^{-1}\left(1-2\mathrm{e}^{-\frac{\eta_\mathrm{B}\left(R_\mathrm{o}\epsilon_{\tilde{n}}+R_\mathrm{t}\right)}{\bar{\gamma}_\mathrm{B}}}\right),\dots,\sqrt{2}\mathrm{erf}^{-1}\left(1-2\mathrm{e}^{-\frac{\eta_\mathrm{B}\left(R_\mathrm{o}\epsilon_{\tilde{n}}+R_\mathrm{t}\right)}{\bar{\gamma}_\mathrm{B}}}\right);\vartheta_{G_{k_\mathrm{B},\tilde{k}_\mathrm{B}}}\right)\label{eq-sop-ray}
\end{align}
\hrulefill
\end{figure*}
\end{corollary}

\begin{corollary}
The SEE for the FAS under correlated Rayleigh fading channels is given by \eqref{eq-see-ray} (see top of next page), in which $\vartheta_{G_{k_i,\tilde{k}_i}}$ is defined in \eqref{eq-g-par}.
\begin{figure*}[]
\begin{align}\notag
\mathcal{E}_\mathrm{s}=& \sum_{\tilde{n}=1}^{\tilde{N}}\frac{w_{\tilde{n}}\mathrm{e}^{\epsilon_{\tilde{n}}}}{\left(1+\epsilon_{\tilde{n}}\right)\left(P/\alpha+P_\mathrm{c}+P_\mathrm{act}\right)}\left[1-\Phi_{\mathbf{R}_{k_\mathrm{B},\tilde{k}_\mathrm{B}}}\left(\sqrt{2}\mathrm{erf}^{-1}\left(1-2\mathrm{e}^{-\frac{\eta_\mathrm{B}\epsilon_{\tilde{n}}}{\bar{\gamma}_\mathrm{B}}}\right),\dots,\sqrt{2}\mathrm{erf}^{-1}\left(1-2\mathrm{e}^{-\frac{\eta_\mathrm{B}\epsilon_{\tilde{n}}}{\bar{\gamma}_\mathrm{B}}}\right);\vartheta_{G_{k_\mathrm{B},\tilde{k}_\mathrm{B}}}\right)\right]\\
&\times \Phi_{\mathbf{R}_{k_\mathrm{E},\tilde{k}_\mathrm{E}}}\left(\sqrt{2}\mathrm{erf}^{-1}\left(1-2\mathrm{e}^{-\frac{\eta_\mathrm{E}\epsilon_{\tilde{n}}}{\bar{\gamma}_\mathrm{E}}}\right),\dots,\sqrt{2}\mathrm{erf}^{-1}\left(1-2\mathrm{e}^{-\frac{\eta_\mathrm{E}\epsilon_{\tilde{n}}}{\bar{\gamma}_\mathrm{E}}}\right);\vartheta_{G_{k_\mathrm{E},\tilde{k}_\mathrm{E}}}\right)\label{eq-see-ray}
\end{align}
\hrulefill
\end{figure*}
\end{corollary}

Furthermore, the behavior of the considered FAS for sufficiently large $W_\mathrm{B}$ and $K_\mathrm{B}$, i.e., $W_\mathrm{B}\rightarrow\infty$ and/or $K_\mathrm{B}\rightarrow\infty$, can provide deeper insights into the derived expressions, which we provide in the following remarks.
\begin{remark}\label{remark4}
Given \eqref{eq-rho}, when $W_\mathrm{B}\rightarrow\infty$ for a fixed value of $K_i$ and $W_\mathrm{E}$, the spherical Bessel function reaches $0$. Consequently, the corresponding correlation matrix $\mathbf{R}_{k_\mathrm{B},\tilde{k}_\mathrm{B}}$ becomes $0$, and the joint CDF of multivariate normal distribution $\Phi_{\mathbf{R}_{k_\mathrm{B},\tilde{k}_\mathrm{B}}}$ in \eqref{eq-asc-ray}, \eqref{eq-sop-ray}, and \eqref{eq-see-ray} becomes a constant value, equal to the product of $K_\mathrm{B}$ marginal normal CDFs due to their independence. Hence, the ASC, SOP, and SEE converge to a constant value for $W_\mathrm{B}\rightarrow\infty$ and $K_\mathrm{B}=\mathrm{constant}$.
\end{remark}
\begin{remark}\label{remark5}
Given \eqref{eq-rho}, when $K_\mathrm{B}\rightarrow\infty$ for a fixed value of $W_i$ and $K_\mathrm{E}$, the spherical Bessel function reaches $1$. As such, the corresponding correlation matrix $\mathbf{R}_{k_\mathrm{B},\tilde{k}_\mathrm{B}}$ becomes $1$ (e.g., without loss of generality, let $\Omega_\mathrm{B}=1$), and the joint CDF of multivariate normal distribution $\Phi_{\mathbf{R}_{k_\mathrm{B},\tilde{k}_\mathrm{B}}}$ in \eqref{eq-asc-ray}, \eqref{eq-sop-ray}, and \eqref{eq-see-ray} becomes a constant value, bounded by the minimum of $K_\mathrm{B}$ marginal normal CDFs due to their full dependency. This upper bound is derived through the Fr\'{e}chet-Hoeffding theorem, which indicates that for any copula $C: \left[0,1\right]^d\rightarrow \left[0,1\right]$, the upper bound holds as $C\left(u_1,\dots,u_d\right)\leq\min\left\{u_1,\dots,u_d\right\}$ \cite[Theorem 2]{ghadi2022capacity}. Hence, the ASC, SOP, and SEE converge to a constant value for $K_\mathrm{B}\rightarrow\infty$ and $W_\mathrm{B}=\mathrm{constant}$. 
\end{remark}

\section{Numerical Results}\label{sec-num}
Here, we present numerical results to evaluate the considered system performance in terms of the ASC, SOP, and SEE which are confirmed by Monte Carlo simulations. We exploit Algorithm 1 in \cite{ghadi2023gaussian} to simulate the Gaussian copula by considering  the Cholesky decomposition of the defined correlation matrix $\mathbf{R}$. In addition, the joint CDF of the multivariate normal distribution can be estimated numerically or implemented via the mathematical package of  various programming languages. Moreover, simulation parameters are given in Table \ref{table1}.

\begin{table}\captionof{table}{Simulation Parameters}\label{table1} \centering
\begin{tabular}{ |p{3.5cm}||p{1.5cm}||p{2.1cm}|  }
	\hline
		\centering \textbf{Definition} &\centering \textbf{Parameter} & \hspace{0.6cm}\textbf{Value} \\
	\hline
Secrecy rate&\centering$R_\mathrm{s}$& $1$ bits \\
	\hline
Circuit power & \centering$P_\mathrm{c}$ & $20$ dBm  \\
	\hline
Port activation power & \centering$P_\mathrm{act}$ & $10$ dBm  \\
\hline
Number of activated ports & \centering$\breve{K}$ & $1$ \\
\hline
	Drain efficiency of HPA &\centering$\alpha$& $1$\\
	\hline
	  Path-loss exponent& \centering$\nu$ &$2.1$  \\
	  \hline
	Scale parameter & \centering$\eta_i$ &$1$\\
	\hline
			Trade-off parameter & \centering$\tilde{N}$ &$2$\\
	\hline
	Noise power at Bob& \centering$\delta^2_\mathrm{B}$ & $-30$ dBm\\
	\hline
	Noise power at Eve& \centering $\delta^2_\mathrm{E}$ & $-20$ dBm\\
	\hline
	Number of ports at Bob& \centering$K_\mathrm{B}$ & $\{4, 16, 25\}$\\
	\hline
		Number of ports at Eve & \centering$K_\mathrm{E}$ & 4\\
		\hline
		FAS size at Bob & \centering$W_\mathrm{B}$ & $\left\{1, 2.25, 9\right\}$  $\lambda^2$\\
	\hline
	FAS size at Eve& \centering$W_\mathrm{E}$ & $1$ $\lambda^2$\\
	\hline
	Number of antennas at Bob& \centering$N_\mathrm{B}$  &$\left\{2, 4, 8, 16\right\}$\\
	\hline
	Number of antennas at Eve& \centering$N_\mathrm{E}$ &$\left\{2, 4\right\}$\\
	\hline
	Alice-to-Eve distance& \centering$d_\mathrm{E}$ &$\left\{1, 3\right\}$ m \\
	\hline
	Alice-to-Bob distance& \centering$d_\mathrm{B}$ &$1$ m \\
	\hline
		    Transmit power& \centering$P$&$-27$ dBm  \\
	\hline
\end{tabular}
\end{table}

Fig.~\ref{fig_1} illustrates the performance of derived secrecy metrics in terms of the average SNR at Bob $\bar{\gamma}_\mathrm{B}$ and transmit power $P$ for given number of fluid antenna ports $K_\mathrm{B}$ and selected value of fluid antenna size $W_\mathrm{B}$. We consider first the balanced case ($\kappa=1$, i.e. $\kappa({\rm dB}) = 0$~dB) on which the eavesdropper channel is not degraded w.r.t. the legitimate one. In Figs.~\ref{fig_c_snr} and \ref{fig_p_g}, as expected, we can see that as $\bar{\gamma}_\mathrm{B}$ grows, the performance of ASC and SOP enhances for a fixed value of $\kappa$, which is reasonable since the quality of main channel (Alice-to-Bob) improves. However, note that the ASC reaches a saturation point after which it is not possible to increase the ASC by increasing the transmit power. In this regime, the capacities for the legitimate and eavesdropper's links are well approximated by $\mathcal{C}_i\approx \log_2\left(\gamma_i\right)$ and the ASC is upper-bounded by $\mathcal{C}_\mathrm{s}\approx \log_2\left(\kappa g_{\mathrm{FAS,B}}/g_{\mathrm{FAS,E}}\right)$ similar to the observation made in \cite{jeon2011bounds}. This is confirmed by the asymptotic ASC value that is reached a $\bar\gamma_\mathrm{B}$ grows. In Fig.~\ref{fig_e_p_k}, it can be observed that the SEE performance initially improves to its maximum value and then becomes worse, increasing the transmit power $P$. The main reason behind such a trend is that the ASC is initially able to overcome the power consumption before reaching the extreme point; however, by continuously increasing $P$, the effect of power consumption becomes more noticeable, and hence, the SEE performance weakens after the extreme point. Moreover, under the constant values of $K_\mathrm{E}$ and $W_\mathrm{E}$ in Figs.~\ref{fig_c_snr}--\ref{fig_e_p_w}, we can observe that considering a larger number of fluid antenna ports (e.g., $K_\mathrm{B}=4\times4$) or a larger value of fluid antenna size (e.g., $W_\mathrm{B}=3\times 3$) at Bob can improve the efficiency of SOP, ASC, and SEE, namely, more secure and reliable transmission. In Figs.~\ref{fig_c_snr}--\ref{fig_e_p_k}, we also consider two benchmarks: 1) both Bob and Eve can be considered the multiple antenna receivers (i.e, SIMO) with $N_i$ antennas utilizing  the MRC technique; and 2) both Bob and Eve exploit AS under the SC technique. In the case of AS, the antenna with the highest SNR is selected, while in SIMO systems, all $N_i$ antennas contribute to signal reception through MRC, optimizing the SNR by combining signals from multiple antennas. For both benchmarks, $N_i$ fixed-position antennas are placed with a separation of half the wavelength, which is a common practice to avoid mutual coupling and ensure independent fading channels. Given the point that only the best port with maximum received SNR is activated in the FAS but all $N_i$ antennas are activated in the SIMO system due to MRC, it can be clearly seen that the FAS with a pre-defined number of fluid antenna ports at Bob (e.g., $K_\mathrm{B}=2\times2$) can significantly provide a higher ASC and SEE and lower SOP compared to the SIMO system with a given number of antennas at Bob (e.g., $N_\mathrm{B}=4$). For instance, under a given SNR $\bar{\gamma}_\mathrm{B}=10$ dB, the ASC is almost equal to $1.5$ bps/Hz and the SOP is around in the order of $10^{-4}$ for the FAS, while these values are near $0.7$ bps/Hz and $10^{-3}$ in the SIMO system. Moreover, assuming a fixed antenna separation of $\lambda/2$, we observe that the FAS performs better compared to the AS that utilizes the SC scheme. Such an improvement is generally achieved due to the capability of fluid antenna in switching to the best port within a finite size $W$, the ability to adapt and reconfigure their antenna patterns dynamically based on the channel conditions through an extreme resolution of fading envelope. Moreover, the main reason of SEE enhancement in the FAS compared with the SIMO system is that activating only one port of FAS needs less power compared to the case of activating all antennas in the SIMO system. 
\begin{figure*}
	\centering
	\hspace{-0.5cm}\subfigure[ASC]{%
		\includegraphics[width=0.35\textwidth]{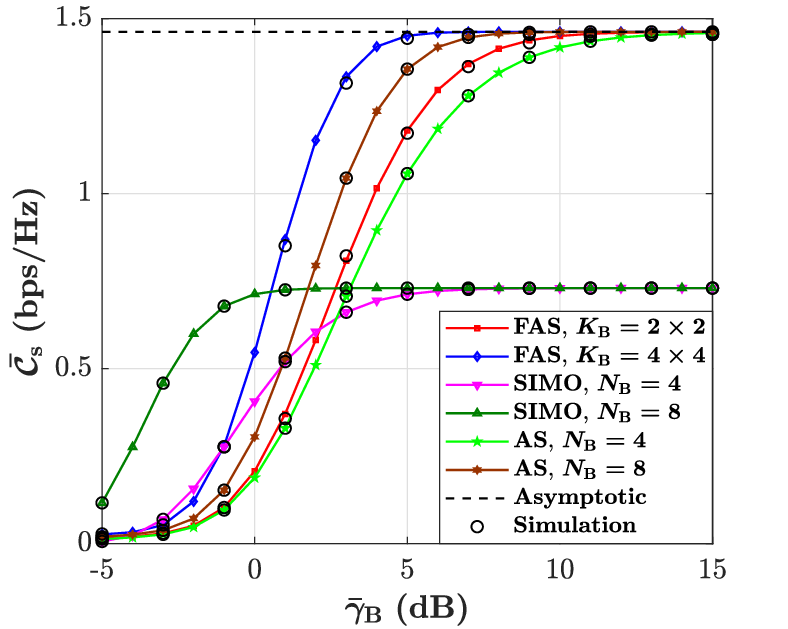}\label{fig_c_snr}%
	}\hspace{-0.3cm}
	\subfigure[SOP]{%
		\includegraphics[width=0.35\textwidth]{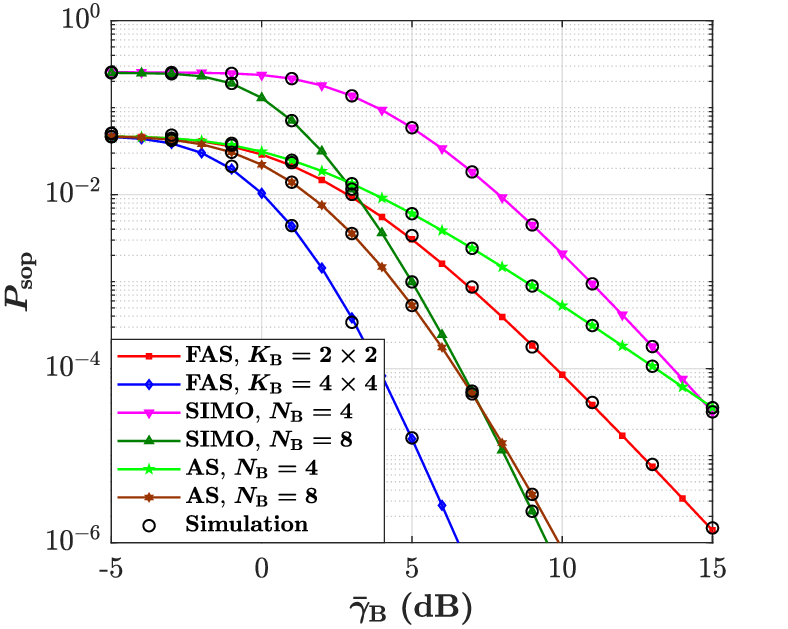}\label{fig_p_g}%
	}\hspace{-0.3cm}
	\subfigure[SEE]{%
		\includegraphics[width=0.35\textwidth]{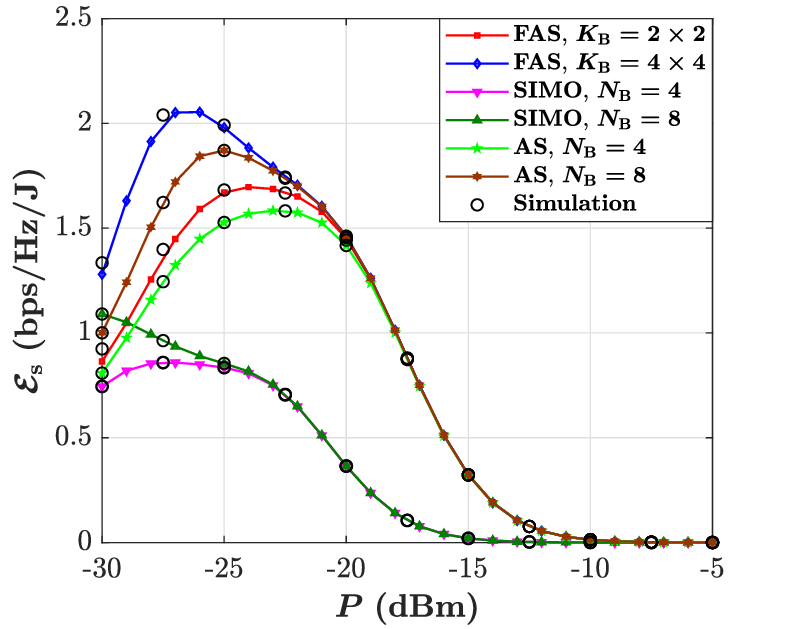}\label{fig_e_p_k}%
	}\hspace{-2cm}\\
	\hspace{-0.5cm}	\centering
	\subfigure[ASC]{%
		\includegraphics[width=0.35\textwidth]{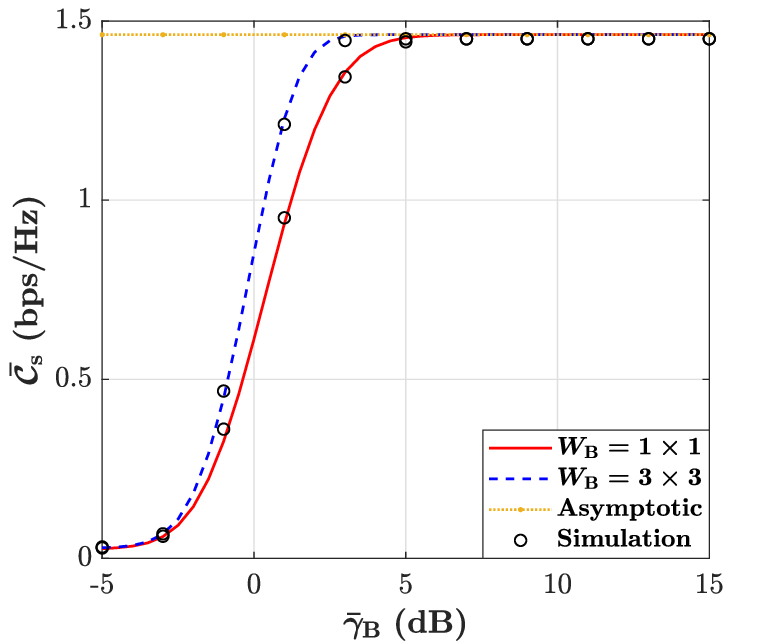}\label{fig_c_snr_w}%
	}\hspace{-0.35cm}
	\subfigure[SOP]{%
		\includegraphics[width=0.35\textwidth]{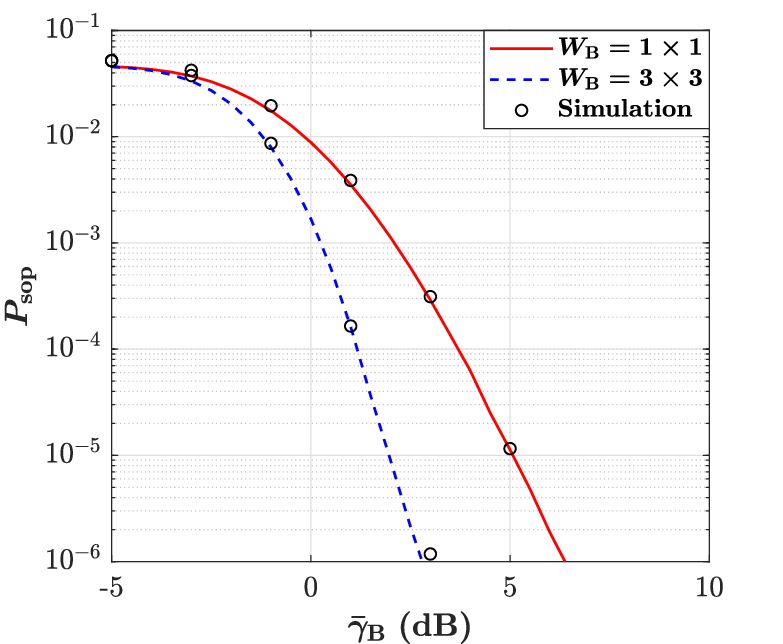}}\label{fig_p_snr_w}\hspace{-0.3cm}
	\subfigure[SEE]{%
		\includegraphics[width=0.35\textwidth]{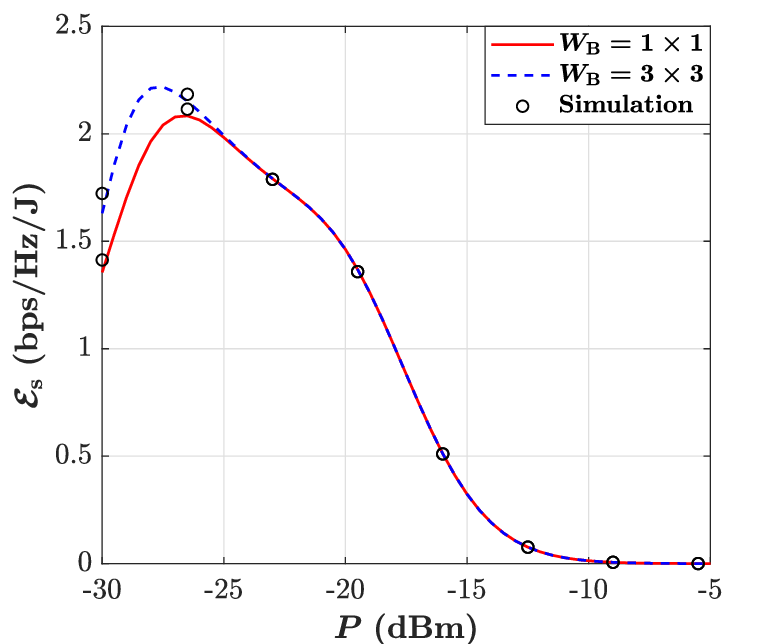}\label{fig_e_p_w}%
	}\hspace{2cm}
\caption{Secrecy metrics versus the average SNR at Bob $\bar{\gamma}_\mathrm{B}$ and transmit power $P$ for selected values of $K_\mathrm{B}$ and $W_\mathrm{B}$. 
}\label{fig_1}\vspace{0cm}
\end{figure*}

To gain more insights into how exactly the number of fluid antenna ports and fluid antenna size affect the secrecy performance, we present the results in Figs.~\ref{fig_2} and \ref{fig_3}, respectively. It is worth pointing out that from the PLS viewpoint, it is classically assumed that the main channel condition is always better than the eavesdropper channel; however, we consider all three possible scenarios relating to the channel condition, i.e., captured by $\kappa>=<1$. In Figs.~\ref{fig_c_k}--\ref{fig_e_k}, under a fixed number of fluid/multiple antenna at Eve, we can observe that the ASC, SOP, and SEE performance  for the FAS initially improves and then converges to a floor as $K_\mathrm{B}$ continuously increases (see Remark \ref{remark4}). This limitation is mainly due to the fact that when increasing $K_\mathrm{B}$ for a given value of $W_\mathrm{B}$, the space between ports decreases and the spatial correlation between them increases as it can be derived in \eqref{eq-rho}; hence, lower diversity gain is achieved until eventually saturating. In addition, in Fig.~\ref{fig_c_k}, we see that the ASC provided by the FAS is almost $0.5$ bps/Hz and $0.4$ bps/Hz higher than that of achieved by the SIMO system and the AS in the advantage ($\kappa>1$) case, respectively. Moreover, in Fig.~\ref{fig_p_k}, under the scenario that the main channel is better than the eavesdropper channel (i.e., $\kappa>1$), it can be found that the SOP of the FAS is nearly $10^{-3}$, while the SOP of SIMO and AS are approximately $3\times10^{-2}$ and $7\times10^{-3}$, respectively. Furthermore, an interesting point is that even under the case that the main channel is worse than the eavesdropper channel (i.e., disadvantage case,  $\kappa<1$), the SOP performance of FAS  marginally improves as $K_\mathrm{B}$ grows and provides lower values compared with the SIMO system. Additionally, the results in Fig.~\ref{fig_e_k} reveal that even if the number of antennas is equal to the number of fluid antenna ports (e.g., $N_\mathrm{B}=K_\mathrm{B}=16$), SEE achieves higher values in FAS, namely, the SEE of fluid antenna is $1.12$ bps/Hz/J higher than that of multiple fixed antennas and $0.17$ bps/Hz/J higher than that of AS. By increasing the fluid antenna size at Bob $W_\mathrm{B}$, it can be observed in Fig.~\ref{fig_c_w} that ASC increases and then plateaus. Such an improvement is mainly because increasing the spatial separation between the fluid antenna ports by increasing $W_\mathrm{B}$ for a fixed $K_\mathrm{B}$ leads to significant reduction in spatial correlation. Nonetheless, as shown in Fig.~\ref{fig_p_w}, the SOP of FAS improves as $W_\mathrm{B}$ raises and then converges to a constant value if $K_\mathrm{B}$ is finite. Moreover, following the trend of ASC, we can observe in Fig.~\ref{fig_e_w} that the SEE performance improves and then becomes constant by increasing $W_\mathrm{B}$ (see Remark \ref{remark5}). In contrast, in all the results in Figs.~\ref{fig_c_w}--\ref{fig_e_w}, it can be seen that the performance of ASC, SOP, and SEE is constant since the diversity gain of SIMO and AS remains the same even if $W_\mathrm{B}$ is increased.
 
To specifically evaluate the large-scale channel effects on system performance, we focus on the path-loss effect for the received SNR at node $i$ through $ d_i^{\nu}$ in \eqref{eq:SNR}, and now setting $\delta_{\rm B}=\delta_{\rm E}$ for convenience. By doing so, the secrecy metrics are directly impacted by the dissimilar path-loss effects. For this purpose, we evaluate the behavior of ASC, SOP, and SEE as functions of $d_\mathrm{B}$ for selected values of $d_\mathrm{E}$, as illustrated in Fig.~\ref{fig-dis}. Following the concept of the secrecy coverage region in \cite{ghadi2021newcopula}, we assume that Eve is located very close to Alice to determine how far Bob can be from Alice while maintaining secure communication. In all three figures, Figs.~\ref{fig-c-d}--\ref{fig-e-d}, we see that the performance of ASC, SOP, and SEE deteriorates as $d_\mathrm{B}$ increases for fixed values of $d_\mathrm{E}={1,3}$. This is reasonable since path-loss degrades the quality of the channel. Furthermore, by carefully observing the curves, we can see that secure communication can still be guaranteed even if the ratio of $d_\mathrm{B}$ to $d_\mathrm{E}$ is equal to or greater than $1$. Additionally, we observe that this performance improves when $K_\mathrm{B}$ and $W_\mathrm{B}$ increase from $4$ to $9$ and from $1\lambda^2$ to $2.25\lambda^2$, respectively.
\begin{figure*}
	\centering
\hspace{-0.5cm}	\subfigure[ASC]{%
		\includegraphics[width=0.35\textwidth]{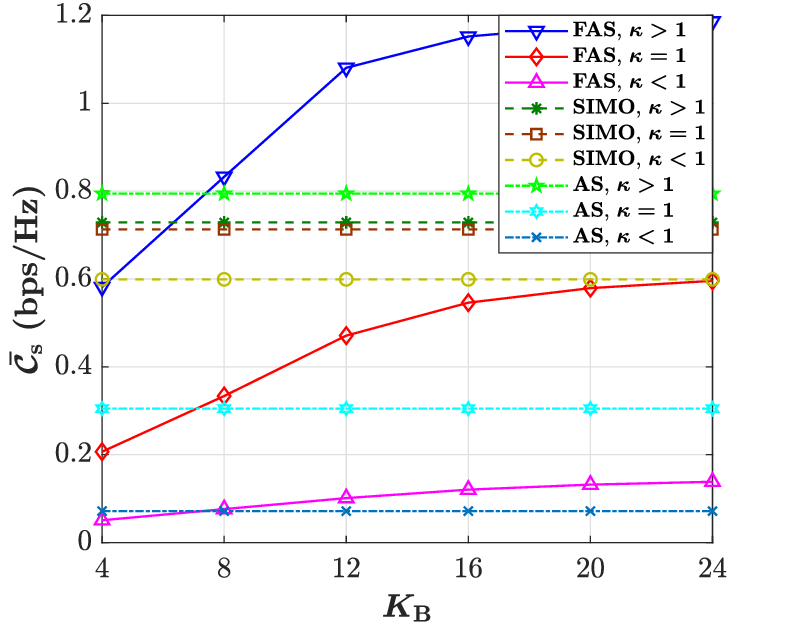}\label{fig_c_k}%
	}\hspace{-0.3cm}
	\subfigure[SOP]{%
		\includegraphics[width=0.35\textwidth]{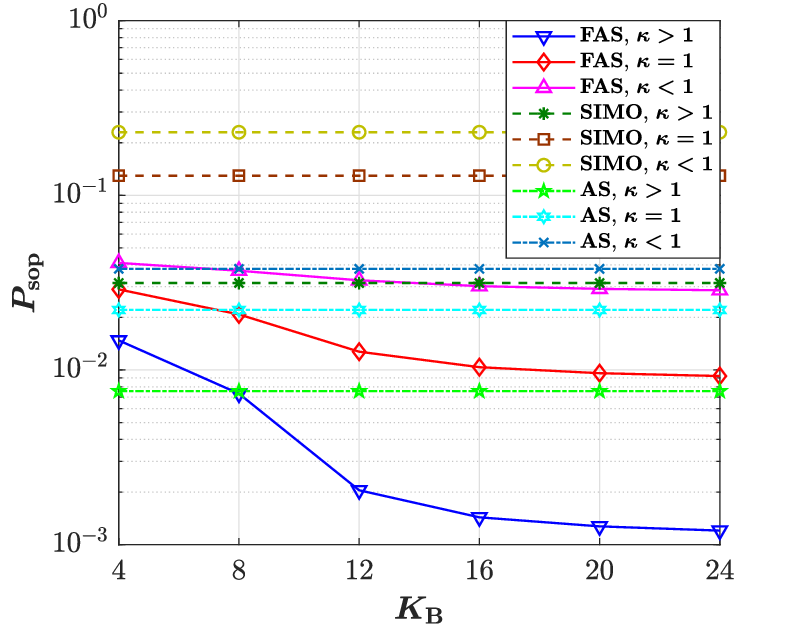}\hspace{-0.5cm}\label{fig_p_k}%
	}
		\subfigure[SEE]{%
		\includegraphics[width=0.35\textwidth]{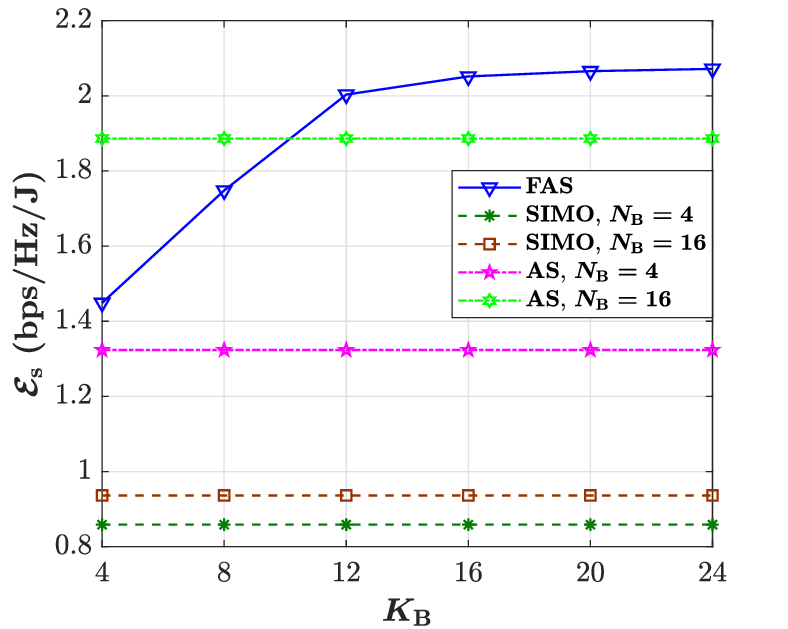}\label{fig_e_k}%
	}
	\caption{Secrecy metrics versus the  number of fluid antenna ports $K_\mathrm{B}$ for different cases of $\kappa$. 
	}\label{fig_2}\vspace{0cm}
\end{figure*}
\begin{figure*}
	\centering
\hspace{-0.5cm}	\subfigure[ASC]{%
		\includegraphics[width=0.35\textwidth]{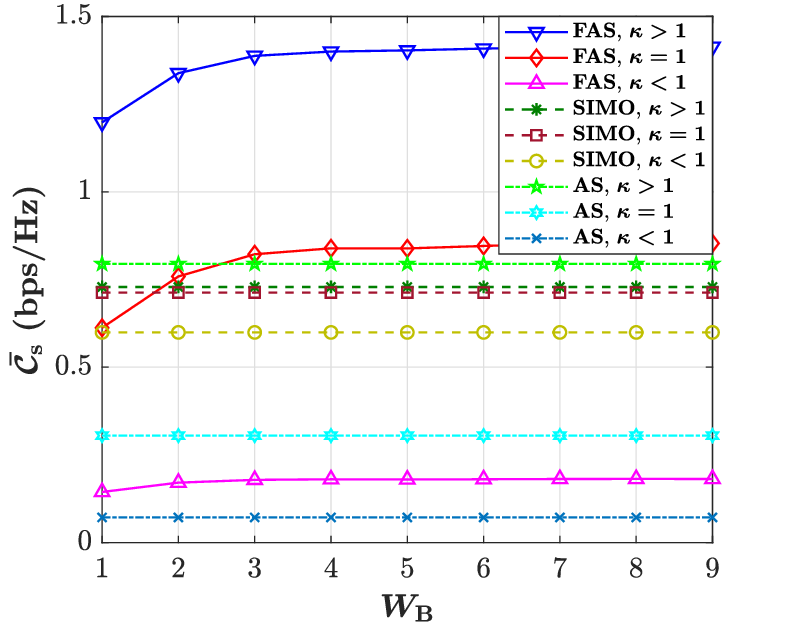}\label{fig_c_w}%
	}\hspace{-0.3cm}
	\subfigure[SOP]{%
		\includegraphics[width=0.35\textwidth]{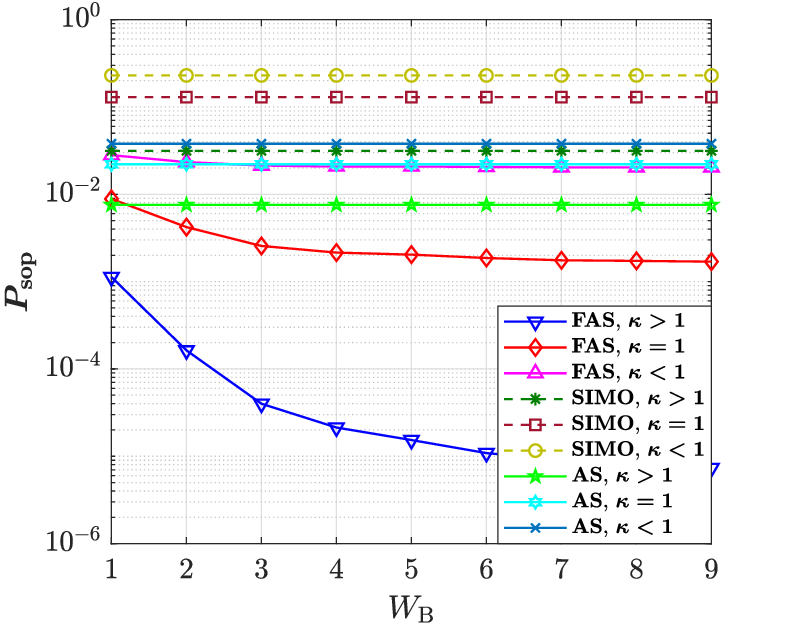}\label{fig_p_w}%
	}\hspace{-0.5 cm}
	\subfigure[SEE]{%
		\includegraphics[width=0.35\textwidth]{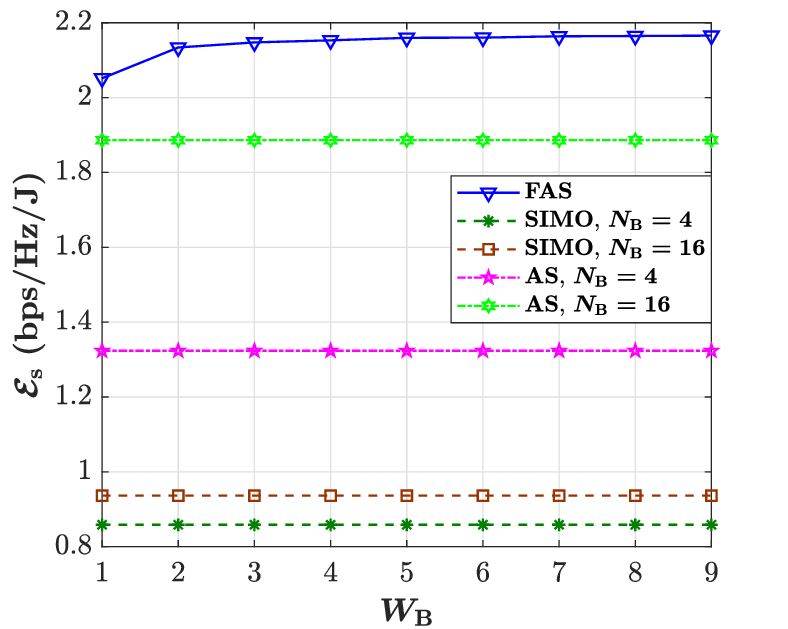}\label{fig_e_w}%
	}
	\caption{Secrecy metrics versus the fluid antenna size $W_\mathrm{B}$ for different cases of $\kappa$. 
	}\label{fig_3}\vspace{0cm}
\end{figure*}

\begin{figure*}
	\centering
	\hspace{-0.5cm}	\subfigure[ASC]{%
		\includegraphics[width=0.35\textwidth]{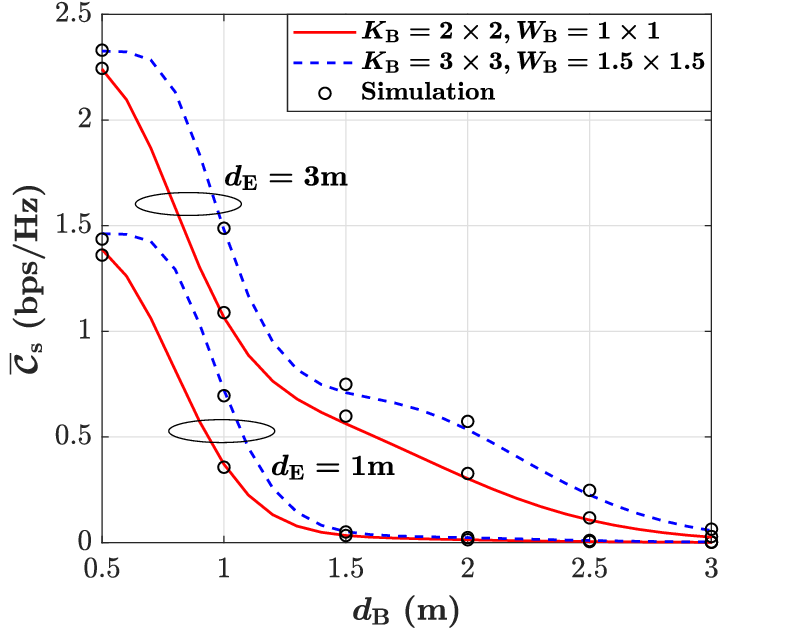}\label{fig-c-d}%
	}\hspace{-0.3cm}
	\subfigure[SOP]{%
		\includegraphics[width=0.35\textwidth]{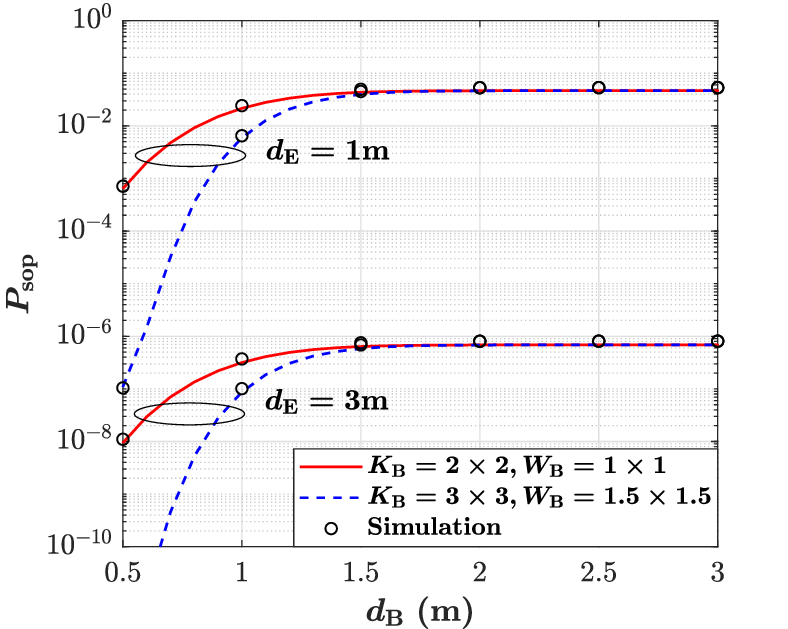}\label{fig-o-d}%
	}\hspace{-0.5 cm}
	\subfigure[SEE]{%
		\includegraphics[width=0.35\textwidth]{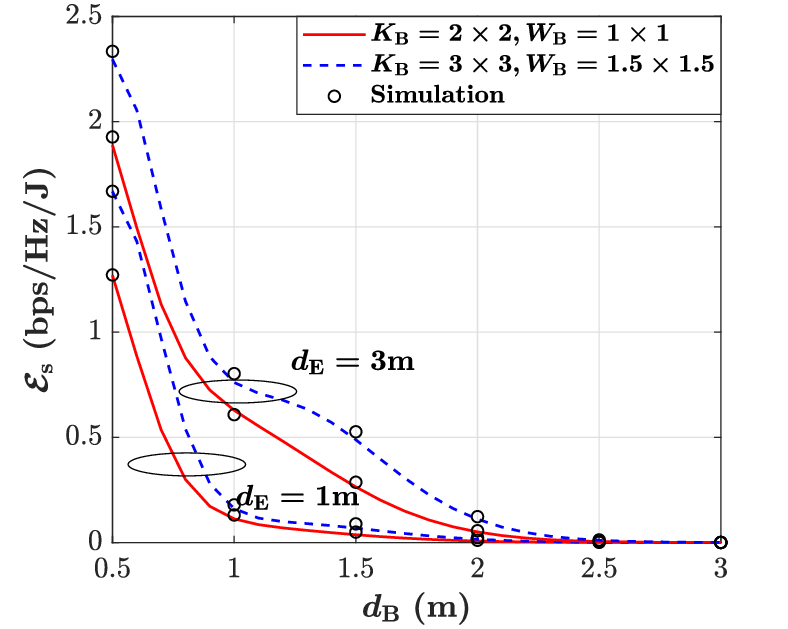}\label{fig-e-d}%
	}
	\caption{Secrecy metrics versus the distance between Alice and Bob $d_\mathrm{B}$ for selected values of $d_\mathrm{E}$, $K_\mathrm{B}$, and $W_\mathrm{B}$.}\label{fig-dis}\vspace{0cm}
\end{figure*}

\section{Conclusion}\label{sec-con}
This paper studied the secrecy performance of FAS-aided communication under arbitrarily correlated fading channels. In particular, we considered that a single fixed-position antenna transmitter sends a confidential message to a legitimate receiver while an eavesdropper attempts to decode the transmitted signal. We also assumed that both the legitimate receiver and the eavesdropper are equipped with the FAS technology. Under these assumptions, by utilizing the copula method, we derived the general distributions of the equivalent channel at both nodes, which are valid for any arbitrarily correlated fading distributions. Then, we obtained general analytical expressions for the ASC, SOP, and SEE by exploiting the Gauss-Laguerre quadrature. Furthermore, to analyze the system's performance, we derived the ASC, SOP, and SEE under correlated Rayleigh fading channels using the Gaussian copula. Additionally, we evaluated the system behavior in the high SNR regime by deriving the asymptotic expressions of the ASC. Finally, by numerically comparing the FAS with the SIMO system and AS, we showed that considering the FAS instead of TAS can remarkably enhance the secrecy performance. \vspace{-0.2cm}
\appendices
\section{Proof of Proposition \ref{pr-asc}}\label{app-asc}
Given the definition of ASC, \eqref{eq-def-asc} can be rewritten as 
\begin{align}
	&\hspace{-1mm}\bar{\mathcal{C}}_\mathrm{s}=\int_0^\infty\hspace{-2mm}\int_0^{\gamma_{{\mathrm{B}}}} \hspace{-2mm}\log_2\left(\frac{1+\gamma_{{\mathrm{B}}}}{1+\gamma_{{\mathrm{E}}}}\right)f_{\gamma_{{\mathrm{B}}}}\left(\gamma_{{\mathrm{B}}}\right)f_{\gamma_{{\mathrm{E}}}}\left(\gamma_{{\mathrm{E}}}\right)d\gamma_{{\mathrm{E}}}d\gamma_{{\mathrm{B}}}\\\notag
	&\hspace{2mm}=\int_0^\infty\hspace{-2mm}\int_0^\infty\log_2\left(1+\gamma_{{\mathrm{B}}}\right)f_{\gamma_{{\mathrm{B}}}}\left(\gamma_{{\mathrm{B}}}\right)f_{\gamma_{{\mathrm{E}}}}\left(\gamma_{{\mathrm{E}}}\right)d\gamma_{{\mathrm{E}}} d\gamma_{{\mathrm{B}}}\\ \notag
	&\quad-\int_0^\infty \log_2\left(1+\gamma_{{\mathrm{B}}}\right)f_{\gamma_{{\mathrm{B}}}}\left(\gamma_{{\mathrm{B}}}\right)\bar{F}_{\gamma_{{\mathrm{E}}}}\left(\gamma_{{\mathrm{B}}}\right)d\gamma_{{\mathrm{B}}} \\
	&\quad-\int_0^\infty \log_2\left(1+\gamma_{{\mathrm{B}}}\right)f_{\gamma_{{\mathrm{E}}}}\left(\gamma_{{\mathrm{B}}}\right)\bar{F}_{\gamma_{{\mathrm{B}}}}\left(\gamma_{{\mathrm{B}}}\right)d\gamma_{{\mathrm{B}}}\\
	&=\frac{1}{\ln 2}\left[\underset{\mathcal{I}_1}{\underbrace{\int_0^\infty \frac{\bar{F}_{\gamma_{{\mathrm{B}}}}\left(\gamma_{{\mathrm{B}}}\right)}{1+\gamma_{{\mathrm{B}}}}d\gamma_{{\mathrm{B}}}}}-\underset{\mathcal{I}_2}{\underbrace{\int_0^\infty \frac{\bar{F}_{\gamma_{{\mathrm{B}}}}\left(\gamma_{{\mathrm{B}}}\right)\bar{F}_{\gamma_{{\mathrm{E}}}}\left(\gamma_{{\mathrm{B}}}\right)}{1+\gamma_{{\mathrm{B}}}}d\gamma_{{\mathrm{B}}}}}\right]\\
	&=\frac{1}{\ln 2}\int_0^\infty \frac{\bar{F}_{\gamma_{{\mathrm{B}}}}\left(\gamma_{{\mathrm{B}}}\right)}{1+\gamma_{{\mathrm{B}}}}F_{\gamma_{{\mathrm{E}}}}\left(\gamma_{{\mathrm{B}}}\right)d\gamma_{{\mathrm{B}}},\label{eq-app1}
\end{align}
in which $\bar{F}_{\gamma_{{\mathrm{B}}}}\left(\gamma_{{\mathrm{B}}}\right)=1-F_{\gamma_{{\mathrm{B}}}}\left(\gamma_{{\mathrm{B}}}\right)$ is the complementary CDF (CCDF) of the SNR $\gamma_{{\mathrm{B}}}$. The integral $\mathcal{I}_1$ mentions the ASC
when Eve is not available and the integral $\mathcal{I}_2$ indicates the ASC loss due to Eve. By inserting \eqref{eq-cdf-gen} into \eqref{eq-app1} and after doing the transformation of random variables, we have
\begin{align}\nonumber
\bar{\mathcal{C}}_\mathrm{s}=&\frac{1}{\ln 2}\int_0^\infty \frac{\left[1-C\left(F_{g_{1,\mathrm{B}}}\left(\frac{\gamma_{{\mathrm{B}}}}{\bar{\gamma}_\mathrm{B}}\right),\dots, F_{g_{k,\mathrm{B}}}\left(\frac{\gamma_{{\mathrm{B}}}}{\bar{\gamma}_\mathrm{B}}\right);\vartheta_C\right)\right]}{1+\gamma_{{\mathrm{B}}}}\\
&\,\times C\left(F_{g_{1,\mathrm{E}}}\left(\frac{\gamma_{{\mathrm{B}}}}{\bar{\gamma}_\mathrm{E}}\right),\dots, F_{g_{k,\mathrm{E}}}\left(\frac{\gamma_{{\mathrm{B}}}}{\bar{\gamma}_\mathrm{E}}\right);\vartheta_C\right)d\gamma_{{\mathrm{B}}}.\label{eq-app2}
\end{align}
Since the provided integral in \eqref{eq-app2} is mathematically intractable to compute, we exploit the Gauss-Laguerre quadrature method which is defined as the
following lemma.

\begin{lemma}\label{lemma1}
	The Gauss-Laguerre quadrature is defined as \cite{abromowitz1972handbook}
	\begin{align}
		\int_0^\infty \mathrm{e}^{-x}\Lambda(x)dx\approx\sum_{\tilde{n}=1}^{\tilde{N}}w_{\tilde{n}}\Lambda\left(\epsilon_{\tilde{n}}\right)
	\end{align}
	in which $\epsilon_{\tilde{n}}$ is the $\tilde{n}$-th root of Laguerre polynomial $L_{\tilde{N}}\left(\epsilon_{\tilde{n}}\right)$ and $w_{\tilde{n}}=\frac{\epsilon_{\tilde{n}}}{2\left(\tilde{N}+1\right)^2L^2_{\tilde{N}+1}\left(\epsilon_{\tilde{n}}\right)}$.
\end{lemma}

Hence, by applying Lemma \ref{lemma1} into \eqref{eq-app2}, \eqref{eq-asc-gen} is derived and the proof is completed.  

\section{The Proof of Proposition \ref{pro-sop} } \label{app-sop}
By substituting the definition of secrecy capacity from \eqref{eq-def-sc} into \eqref{eq-def-sop}, the SOP can re-expressed as
\begin{align}
	&P_\mathrm{sop}=\Pr\left(\log_2\left(\frac{1+\gamma_\mathrm{B}}{1+\gamma_\mathrm{E}}\right)\leq R_\mathrm{s}\right)\\
	&=\Pr\left(\gamma_\mathrm{B}\leq 2^{R_\mathrm{s}}\gamma_\mathrm{E}+2^{R_\mathrm{s}}-1\right)\\
	&=\int_0^\infty F_{\gamma_\mathrm{B}}\left(\gamma_\mathrm{o}\right)f_{\gamma_\mathrm{E}}\left(\gamma_\mathrm{E}\right)d\gamma_\mathrm{E}\\ \notag
	&\overset{(b)}{=}\prod_{k=1}^{K}\int_0^\infty \frac{f_{g_{k,\mathrm{E}}}\left(\frac{\gamma_{{\mathrm{E}}}}{\bar{\gamma}_\mathrm{E}}\right)}{\bar{\gamma}_\mathrm{E}^K}c\left(F_{g_{1,\mathrm{E}}}\left(\frac{\gamma_{{\mathrm{E}}}}{\bar{\gamma}_\mathrm{E}}\right),\dots, F_{g_{k,\mathrm{E}}}\left(\frac{\gamma_{{\mathrm{E}}}}{\bar{\gamma}_\mathrm{E}}\right)\right)
	 \\
	&\times C\left(F_{g_{1,\mathrm{B}}}\left(\frac{R_\mathrm{o}\gamma_\mathrm{E}+R_\mathrm{t}}{\bar{\gamma}_\mathrm{B}}\right),\dots, F_{g_{k,\mathrm{B}}}\left(\frac{R_\mathrm{o}\gamma_\mathrm{E}+R_\mathrm{t}}{\bar{\gamma}_\mathrm{B}}\right)\right)d\gamma_{{\mathrm{E}}},\label{eq-sop2}
\end{align}
where $\gamma_\mathrm{o}=R_\mathrm{o}\gamma_\mathrm{E}+R_\mathrm{t}$, $R_\mathrm{o}=2^{R_{\mathrm{s}}}$, and $R_\mathrm{t}=R_\mathrm{o}-1$. Besides, $(b)$ is obtained by transformation of random variables and Theorem \ref{thm-dist-gen}. Then, by applying the Gauss-Laguerre quadrature from Lemma \ref{lemma1} into \eqref{eq-sop2}, the proof is accomplished. 
\section{Construction of the Gaussian Copula } \label{app-g}
Let $\mathbf{s}=\left[S_1,\dots,S_d\right]$ be a vector of $d$ random variables with marginal CDFs $F_{S_q}\left(s_q\right)$ for $q\in\left\{1,\dots,d\right\}$ respectively. Then, by transforming them to uniform variables $U_q$, we have
\begin{align}
	U_q=F_{S_q}\left(s_q\right).
\end{align}
Next, we transform these uniform variables to standard normal variables by using the inverse CDF of the standard normal distribution, i.e.,
\begin{align}
	V_q=\varphi^{-1}\left(U_q\right).
\end{align}
Further, we assume that the transformed variables $V_q$ follow a multivariate normal distribution with correlation matrix $\mathbf{R}$. Hence, we have
$\mathbf{v}=\left[V_1,\dots,V_d\right]\sim\mathcal{N}\left(\mathbf{0},\mathbf{R}\right)$ such that the joint CDF of $\mathbf{v}$ is given by
\begin{align}\label{eq-app-jcdf}
	F_{\mathbf{v}}\left(v\right)=\Phi_{\mathbf{R}}\left(v_1,\dots,v_d\right).
\end{align}
Now, using the copula definition from \eqref{eq-cop-def}, the Gaussian copula $C_\mathrm{G}$, which is the joint CDF of $\mathbf{v}$ evaluated at $\phi^{-1}\left(u_q\right)$, can be derived as 
\begin{align}\notag
	&C_\mathrm{G}\left(u_1,\dots,u_d\right)\\
	&=\Pr\left(V_1\leq\varphi^{-1}\left(u_1\right),\dots,V_d\leq\varphi^{-1}\left(u_d\right);\vartheta_\mathrm{G}\right).
\end{align}
Eventually, using the multivariate normal CDF from \eqref{eq-app-jcdf} and Sklar's theorem, we have \eqref{eq-gc-def1} and the construction of $C_\mathrm{G}$ is completed.

As for the construction of $c_\mathrm{G}$, we can let $\boldsymbol{\varphi}^{-1}=\left[\varphi^{-1}\left(u_1\right),\dots,\varphi^{-1}\left(u_d\right)\right]^T$, and then by using the chain rule, $c_\mathrm{G}$ can be written as 
	\begin{align}
	&c_\mathrm{G}\left(u_1,\dots,u_d;\vartheta_\mathrm{G}\right)\notag\\
	&=\frac{\partial^d C_\mathrm{G}\left(u_1,\dots,u_d;\vartheta_\mathrm{G}\right)}{\partial u_1\dots\partial u_d}\\
	&=\frac{\phi_\mathbf{R}\left(\varphi^{-1}(u_1),\dots,\varphi^{-1}(u_d);\vartheta_\mathrm{G}\right)}{\phi\left(\varphi^{-1}\left(u_1\right)\right)\dots \phi\left(\varphi^{-1}\left(u_1\right)\right)}\\
	&=\left(2\pi\right)^{-\frac{d}{2}}\det\left(\mathbf{R}\right)^{-\frac{1}{2}}\frac{\exp\left(-\frac{1}{2}\left(\boldsymbol{\varphi}^{-1}\right)^T\mathbf{R}^{-1}\boldsymbol{\varphi}^{-1}\right)}{\prod_{q=1}^d\left(2\pi\right)^{-\frac{1}{2}}\exp\left(-\frac{\boldsymbol{\varphi}^{-1}}{2}\right)}, \label{eq-app-cg}
	\end{align}
in which $\phi_\mathrm{R}$ and $\phi$ represent the joint PDF of multivariate normal distribution and the PDF of the univariate normal distribution, respectively. After performing  simplifications in \eqref{eq-app-cg}, the construction of $c_\mathrm{G}$ is accomplished.
\bibliographystyle{IEEEtran}

\begin{thebibliography}{10}
\providecommand{\url}[1]{#1}
\csname url@samestyle\endcsname
\providecommand{\newblock}{\relax}
\providecommand{\bibinfo}[2]{#2}
\providecommand{\BIBentrySTDinterwordspacing}{\spaceskip=0pt\relax}
\providecommand{\BIBentryALTinterwordstretchfactor}{4}
\providecommand{\BIBentryALTinterwordspacing}{\spaceskip=\fontdimen2\font plus
\BIBentryALTinterwordstretchfactor\fontdimen3\font minus
  \fontdimen4\font\relax}
\providecommand{\BIBforeignlanguage}[2]{{%
\expandafter\ifx\csname l@#1\endcsname\relax
\typeout{** WARNING: IEEEtran.bst: No hyphenation pattern has been}%
\typeout{** loaded for the language `#1'. Using the pattern for}%
\typeout{** the default language instead.}%
\else
\language=\csname l@#1\endcsname
\fi
#2}}
\providecommand{\BIBdecl}{\relax}
\BIBdecl

\bibitem{tariq2020speculative}
F.~Tariq {\em et al.}, ``A speculative study on 6G,'' \emph{IEEE Wireless Commun.}, vol.~27, no.~4, pp. 118--125, Aug. 2020.

\bibitem{zhang2019cell}
J.~Zhang {\em et al.}, ``Cell-free massive MIMO: A new next-generation paradigm,'' \emph{IEEE Access}, vol.~7, pp. 99\,878--99\,888, Jul. 2019.

\bibitem{wong2020fluid}
K.-K. Wong, A.~Shojaeifard, K.-F. Tong, and Y.~Zhang, ``Fluid antenna systems,'' \emph{IEEE Trans. Wireless Commun.}, vol.~20, no.~3, pp. 1950--1962, Mar. 2021.

\bibitem{huang2021liquid}
Y.~Huang, L.~Xing, C.~Song, S.~Wang, and F.~Elhouni, ``Liquid antennas: Past, present and future,'' \emph{IEEE Open J. Antennas \& Propag.}, vol.~2, pp. 473--487, Mar. 2021.

\bibitem{song2013efficient}
S.~Song and R.~D. Murch, ``An efficient approach for optimizing frequency reconfigurable pixel antennas using genetic algorithms,'' \emph{IEEE Trans. Antennas \& Propag.}, vol.~62, no.~2, pp. 609--620, Feb. 2014.
\bibitem{Shen-tap_submit2024}
Y. Shen {\em et al.}, ``Design and implementation of mmWave surface wave enabled fluid antennas and experimental results for fluid antenna multiple access,'' {\em arXiv preprint}, \url{arXiv:2405.09663}, May 2024.
\bibitem{Zhang-pFAS2024}
J.~Zhang {\em et al.}, ``A pixel-based reconfigurable antenna design for fluid antenna systems,'' \emph{arXiv preprint}, \url{arXiv:2406.05499}, Jun. 2024.

\bibitem{wong2022bruce}
K.-K. Wong, K.-F. Tong, Y.~Shen, Y.~Chen, and Y.~Zhang, ``Bruce Lee-inspired fluid antenna system: Six research topics and the potentials for 6G,'' \emph{Frontiers Commun. Netw.}, vol.~3, p. 853416, Mar. 2022.

\bibitem{wong2022closed}
K.~Wong, K.~Tong, Y.~Chen, and Y.~Zhang, ``Closed-form expressions for spatial correlation parameters for performance analysis of fluid antenna systems,'' \emph{Elect. Lett.}, vol.~58, no.~11, pp. 454--457, Apr. 2022.

\bibitem{khammassi2023new}
M.~Khammassi, A.~Kammoun, and M.-S. Alouini, ``A new analytical approximation of the fluid antenna system channel,'' \emph{IEEE Trans. Wireless Commun.}, vol. 22, no. 12, pp. 8843--8858, Dec. 2023.

\bibitem{ghadi2023copula}
F.~R. Ghadi, K.-K. Wong, F.~J. Lopez-Martinez, and K.-F. Tong, ``Copula-based performance analysis for fluid antenna systems under arbitrary fading channels,'' \emph{IEEE Commun. Lett.}, vol.~27, no.~11, pp. 3068--3072, Nov. 2023.

\bibitem{ghadi2023gaussian}
F.~R. Ghadi {\em et al.}, ``A Gaussian copula approach to the performance analysis of fluid antenna systems,'' \emph{IEEE Trans. Wireless Commun.}, early access, \url{doi.org/10.1109/TWC.2024.3454558}, Sep. 2024.

{\em arXiv preprint}, \url{arXiv:2309.07506}, 2023.

\bibitem{ma2024covert} R. Ma {\em et al.}, ``Covert mmWave communications with finite blocklength against spatially random wardens,'' \emph{IEEE Internet Things J.}, vol. 11, no. 2, pp. 3402--3416, Jan. 2024.

\bibitem{chen2016survey}
X.~Chen, D.~W.~K. Ng, W.~H. Gerstacker, and H.-H. Chen, ``A survey on multiple-antenna techniques for physical layer security,'' \emph{IEEE Commun. Surv. \& Tut.}, vol.~19, no.~2, pp. 1027--1053, Secondquarter 2017.

\bibitem{xiao2023array}
Z.~Xiao, L.~Zhu, L.~Bai, and X.-G. Xia, \emph{Array beamforming enabled wireless communications}, CRC Press, 2023.

\bibitem{wong2021fluid}
K.-K. Wong and K.-F. Tong, ``Fluid antenna multiple access,'' \emph{IEEE Trans. Wireless Commun.}, vol.~21, no.~7, pp. 4801--4815, Jul. 2022.

\bibitem{wong2022fast}
K.-K. Wong, K.-F. Tong, Y.~Chen, and Y.~Zhang, ``Fast fluid antenna multiple access enabling massive connectivity,'' \emph{IEEE Commun. Lett.}, vol.~27, no.~2, pp. 711--715, Feb. 2023.

\bibitem{wong2023slow}
K.-K. Wong, D.~Morales-Jimenez, K.-F. Tong, and C.-B. Chae, ``Slow fluid antenna multiple access,'' \emph{IEEE Trans. Commun.}, vol. 71, no. 5, pp. 2831--2846, May 2023.

\bibitem{ghadi2023fluid}
F.~R. Ghadi {\em et al.}, ``Fluid antenna-assisted dirty multiple access channels over composite fading,'' \emph{IEEE Commun. Lett.}, vol. 28, no. 2, pp. 382--386, Feb. 2024.

\bibitem{xu2023outage}
H.~Xu {\em et al.}, ``Revisiting outage probability analysis for two-user fluid antenna multiple access system,'' \emph{IEEE Trans. Wireless Commun.}, early access, \url{doi:10.1109/TWC.2024.3363499}, Feb. 2024.

\bibitem{xu2023channel}
H.~Xu {\em et al.}, ``Channel estimation for FAS-assisted multiuser mmWave systems,'' \emph{IEEE Commun. Lett.}, vol. 28, no. 3, pp. 632--636, Mar. 2024.

\bibitem{xu2023capacity}
H.~Xu {\em et al.}, ``Capacity maximization for FAS-assisted multiple access channels,'' {\em arXiv preprint}, \url{arXiv:2311.11037}, 2023.

\bibitem{chen2023energy}
Y.~Chen, S.~Li, Y.~Hou, and X.~Tao, ``Energy-efficiency optimization for slow fluid antenna multiple access using mean-field game,'' \emph{IEEE Wireless Commun. Lett.}, vol. 13, no. 4, pp. 915--918, Apr. 2024.

\bibitem{xu2023energy}
Y.~Xu {\em et al.}, ``Energy efficiency maximization under delay-outage probability constraints using fluid antenna systems,'' in \emph{Proc. IEEE Statistical Signal Process. Workshop (SSP)}, pp. 105--109, 2-5 Jul. 2023, Hanoi, Vietnam.

\bibitem{wong2023compact}
K.-K. Wong, C.-B. Chae, and K.-F. Tong, ``Compact ultra massive antenna array: A simple open-loop massive connectivity scheme,'' \emph{IEEE Trans. Wireless Commun.}, vol. 23, no. 6279--6294, Jun. 2024. 
\bibitem{Hu2024fluid}
G. Hu {\em et al.}, ``Fluid antennas-enabled multiuser uplink: A low-complexity gradient descent for total transmit power minimization,'' \emph{IEEE Commun. Lett.}, vol. 28, no. 3, pp. 602--606, Mar. 2024.

\bibitem{tlebaldiyeva2022enhancing}
L.~Tlebaldiyeva, G.~Nauryzbayev, S.~Arzykulov, A.~Eltawil, and T.~Tsiftsis, ``Enhancing QoS through fluid antenna systems over correlated Nakagami-$m$ fading channels,'' in \emph{Proc. IEEE Wireless Commun. Netw. Conf. (WCNC)}, pp. 78--83, 10-13 Apr. 2022, Austin, TX, USA.

\bibitem{vega2023simple}
J.~D. Vega-S{\'a}nchez, L.~Urquiza-Aguiar, M.~C.~P. Paredes, and D.~P.~M. Osorio, ``A simple method for the performance analysis of fluid antenna systems under correlated Nakagami-$m$ fading,'' \emph{IEEE Wireless Commun. Lett.}, 2023.

\bibitem{new2023fluid}
W.~K. New, K.-K. Wong, H.~Xu, K.-F. Tong, and C.-B. Chae, ``Fluid antenna system: New insights on outage probability and diversity gain,'' \emph{IEEE Trans. Wireless Commun.}, vol. 23, no. 1, pp. 128--140, Jan. 2024.

\bibitem{ghadi2024performance}
F.~R. Ghadi, M.~Kaveh, and K.-K. Wong, ``Performance analysis of fluid antenna-aided backscatter communications systems,'' \emph{IEEE Wirel. Commun. Lett.}, vol. 13, no. 9, pp. 2412--2416, Sept. 2024.

\bibitem{chai2022port}
Z.~Chai, K.-K. Wong, K.-F. Tong, Y.~Chen, and Y.~Zhang, ``Port selection for fluid antenna systems,'' \emph{IEEE Commun. Lett.}, vol.~26, no.~5, pp. 1180--1184, May 2022.

\bibitem{tang2023fluid}
B.~Tang {\em et al.}, ``Fluid antenna enabling secret communications,'' \emph{IEEE Commun. Lett.}, vol. 27, no. 6, pp. 1491--1495, Jun. 2023.

\bibitem{cheng2023enabling}
Z.~Cheng {\em et al.}, ``Enabling secure wireless communications via movable antennas,'' {\em arXiv preprint}, \url{arXiv:2312.14018}, 2023.

\bibitem{hu2023secure}
G.~Hu, Q.~Wu, K.~Xu, J.~Si, and N.~Al-Dhahir, ``Secure wireless communication via movable-antenna array,'' {\em arXiv preprint}, \url{arXiv:2311.07104}, 2023.

\bibitem{new2023information}
W.~K. New, K.-K. Wong, H.~Xu, K.-F. Tong, and C.-B. Chae, ``An information-theoretic characterization of MIMO-FAS: Optimization, diversity-multiplexing tradeoff and $q$-outage capacity,'' \emph{IEEE Trans. Wireless Commun.}, vol. 23, no. 6, pp. 5541--5556, Jun. 2024. 
\bibitem{ghadi2024ris}
F. R. Ghadi {\em et al.}, ``On Performance of RIS-aided fluid antenna systems,'' \emph{IEEE Wirel. Commun. Lett.}, vol. 13, no. 8, pp. 2175--2179, Aug. 2024.
\bibitem{zhu2024historical}
L.~Zhu and K. K. Wong, ``Historical review of fluid antenna and movable antenna,'' {\em arXiv preprint}, \url{arXiv:2401.02362v2}, 2024.
\bibitem{zhu2023modeling}
L. Zhu, W. Ma, and R. Zhang, ``Modeling and performance analysis for movable antenna enabled wireless communications,'' \emph{IEEE Trans. Wireless Commun.}, vol. 23, no. 6, pp. 6234--6250, Jun. 2024.
\bibitem{ma2023mimo}
W. Ma, L. Zhu, and R. Zhang, ``MIMO capacity characterization for movable antenna systems,'' \emph{IEEE Trans. Wireless Commun.}, vol. 23, no. 4, pp. 3392--3407, Apr. 2024.
\bibitem{wang2024movable}
H. Wang, Q. Wu, and W. Chen, ``Movable antenna enabled interference network: Joint antenna position and beamforming design,'' {\em arXiv preprint}, \url{arXiv:2403.13573}, Mar. 2024.
\bibitem{gao2024joint}
Y. Gao, Q. Wu, and W. Chen, ``Joint transmitter and receiver design for movable antenna enhanced multicast communications,'' {\em arXiv preprint}, \url{arXiv:2404.11881}, Apr. 2024.

\bibitem{new2023fluidnew}
W.~K. New {\em et al.}, ``Fluid antenna system enhancing orthogonal and non-orthogonal multiple access,'' \emph{IEEE Commun. Lett.}, vol. 28, no. 1, pp. 218--222, Jan. 2024.

\bibitem{jeon2011bounds}
H.~Jeon, N.~Kim, J.~Choi, H.~Lee, and J.~Ha, ``Bounds on secrecy capacity over correlated ergodic fading channels at high SNR,'' \emph{IEEE Trans. Inf. Theory}, vol.~57, no.~4, pp. 1975--1983, 2011.

\bibitem{nelsen2006introduction}
R.~B. Nelsen, \emph{An introduction to copulas}, Springer, 2006.

\bibitem{abromowitz1972handbook}
M.~Abromowitz and I.~A. Stegun, ``{\em Handbook of mathematical functions},'' 1972.

\bibitem{ghadi2023newperformance}
F.~R. Ghadi, M.~Kaveh, and D.~Mart{\'\i}n, ``Performance analysis of RIS/STAR-IOS-aided V2V NOMA/OMA communications over composite fading channels,'' \emph{IEEE Trans. Intelligent Veh.}, 2023.

\bibitem{ghadi2020copula}
F.~R. Ghadi and G.~A. Hodtani, ``Copula-based analysis of physical layer security performances over correlated Rayleigh fading channels,'' \emph{IEEE Trans. Inf. Forensics and Security}, vol.~16, pp. 431--440, Aug. 2020.

\bibitem{gholizadeh2015capacity}
M.~H. Gholizadeh, H.~Amindavar, and J.~A. Ritcey, ``On the capacity of MIMO correlated Nakagami-$m$ fading channels using copula,'' \emph{EURASIP J. Wireless Commun. Netw.}, vol. 2015, no.~1, pp. 1--11, 2015.

\bibitem{ghadi2022capacity}
F.~R. Ghadi, F.~J. Martin-Vega, and F.~J. L{\'o}pez-Mart{\'\i}nez, ``Capacity of backscatter communication under arbitrary fading dependence,'' \emph{IEEE Trans. Veh. Technol.}, vol.~71, no.~5, pp. 5593--5598, May 2022.

\bibitem{zheng2019copula}
C.~Zheng, M.~Egan, L.~Clavier, G.~W. Peters, and J.-M. Gorce, ``Copula-based interference models for IoT wireless networks,'' in \emph{Proc. IEEE Int. Conf. Commun. (ICC)}, pp. 1--6, 20-24 May 2019, Shanghai, China.

\bibitem{ghadi2020copula1}
F.~R. Ghadi and G.~A. Hodtani, ``Copula based performance analysis for one-hoping relay channel in wireless ad hoc network with correlated fading channels,'' \emph{IET Signal Process.}, vol.~14, no.~8, pp. 551--559, 2020.

\bibitem{ghadi2022performance}
F.~R. Ghadi and W.-P. Zhu, ``Performance analysis over correlated/independent Fisher-Snedecor $\mathcal{F}$ fading multiple access channels,'' \emph{IEEE Trans. Veh. Technol.}, vol.~71, no.~7, pp. 7561--7571, Jul. 2022.

\bibitem{choi2015copula}
S.~Choi, H.~He, and P.~K. Varshney, ``Copula based dependence modeling for inference in RADAR systems,'' in \emph{Proc. IEEE Radar Conf.}, pp. 197--202, 27-30 Oct. 2015, Johannesburg, South Africa.

\bibitem{ghadi2022impact}
F.~R. Ghadi, F.~J. L{\'o}pez-Mart{\'\i}nez, W.-P. Zhu, and J.-M. Gorce, ``The impact of side information on physical layer security under correlated fading channels,'' \emph{IEEE Trans. Inf. Forensics and Security}, vol.~17, pp. 3626--3636, Oct. 2022.

\bibitem{ghadi2024newperformance}
F.~R. Ghadi and F.~J. L{\'o}pez-Mart{\'\i}nez, ``Performance analysis of SWIPT relay networks over arbitrary dependent fading channels,'' \emph{IEEE Trans. Commun.}, vol. 72, no. 6, pp. 3651--3663, June 2024.
\bibitem{ghadi2021newcopula}
F.~R. Ghadi and G. A. Hodtani, ``Copula-based analysis of physical layer security performances over correlated Rayleigh fading channels,'' \emph{IEEE Trans. Inf. Forensics Secur.}, vol. 16, pp. 431--440, 2021.
\end{thebibliography}


\end{document}